\documentclass[journal]{IEEEtran}
\usepackage[utf8]{inputenc}
\usepackage{bm}
\usepackage{amsmath,epsfig,amssymb,dsfont,amsthm}
\usepackage[margin=1in]{geometry}
\usepackage{hyperref}
\usepackage{parskip}
\usepackage{verbatim}
\usepackage{graphicx}
\usepackage{caption}
\usepackage{subcaption}

\usepackage{graphicx}
	\graphicspath{{./figures/}}
\usepackage{float}
\usepackage{url}
\usepackage{eurosym}
\usepackage{setspace}
\usepackage{balance}
\usepackage{dsfont}
\usepackage{bbm}
\usepackage{array}

\usepackage{multirow}

\theoremstyle{plain}
\newtheorem{Thm}{Theorem}
\newtheorem{Cor}{Corollary}
\newtheorem{Prp}{Proposition}
\newtheorem{Lem}{Lemma}
\newtheorem{Asm}{Assumption}
\newtheorem{Def}{Definition}

\theoremstyle{remark}

\newtheorem{Exa}{Example}

\usepackage{xcolor}

\newcommand{\qedsymb}{\hfill\ensuremath{\blacksquare}}                 

\title{Self-aware Social Learning over Graphs
}
\author{Konstantinos Ntemos, Virginia Bordignon, Stefan Vlaski, and Ali H. Sayed
\thanks{Konstantinos Ntemos, Virginia Bordignon, and Ali H. Sayed are with the School of Engineering, Ecole Polytechnique F\'ed\'erale de Lausanne (EPFL). Stefan Vlaski is with Faculty of Engineering, Department of Electrical and Electronic Engineering, Imperial College London. E-mails: konstantinos.ntemos@epfl.ch, virginia.bordignon@epfl.ch, s.vlaski@imperial.ac.uk,  ali.sayed@epfl.ch.\\
This work was supported in part by the Swiss National Science Foundation
grant 205121-184999.
}}
                
\allowdisplaybreaks
 
\begin{document}

\maketitle
\begin{abstract}
    In this paper we study the problem of social learning under multiple true hypotheses and {\em self-interested} agents which exchange information over a graph. In this setup, each agent receives data that might be generated from a different hypothesis (or state) than the data other agents receive. In contrast to the related literature in social learning, which focuses on showing that the network achieves consensus, here we study the case where every agent is self-interested and wants to find the hypothesis that generates its own observations. However, agents do not know which ones of their peers wants to find the same state with them and as a result they do not know which agents they should cooperate with. To this end, we propose a scheme with {\em adaptive} combination weights and study the consistency of the agents' learning process. The scheme allows each agent to identify and collaborate with neighbors that observe the same hypothesis, while excluding  others, thus resulting in improved performance compared to both non-cooperative learning and cooperative social learning solutions. We analyze the asymptotic behavior of agents' beliefs under the proposed social learning algorithm and provide sufficient conditions that enable all agents to correctly identify their true hypotheses. The theoretical analysis is corroborated by numerical simulations.
\end{abstract}
\begin{IEEEkeywords}
social learning, self-interested agents, information diffusion.
\end{IEEEkeywords}
\section{Introduction}
{\em Social learning}  \cite{lalitha2014social, lalitha2018social,nedic2017fast, zhao2012learning,bordignon2020social, jadbabaie2012non,molavi2018theory,mitra2020new,ntemos2021social} refers to the distributed hypothesis testing problem where agents exchange information over a graph and aim at {\em learning} an unknown hypothesis of interest. Every agent has access to its own data (observations), as well as to information provided by its neighbors. Furthermore, every agent has access to some {\em likelihood functions} that provide the probability of every observation being generated from every possible hypothesis. Under the social learning paradigm every agent utilizes its own observations along with their likelihood functions to perform a Bayesian update of its {\em belief vector} (probability distribution over the possible hypotheses). Moreover, every agent uses a fusion rule (e.g., linear rule \cite{zhao2012learning, jadbabaie2012non}, log-linear rule \cite{lalitha2014social, lalitha2018social, nedic2017fast}) to incorporate the information the exchanged belief vectors from its neighbors. 
By using this procedure and under some  assumptions every agent's beliefs converge to the hypothesis that better explains all agents' models (i.e., likelihood functions). In this way, network consensus is achieved (i.e., all agents' beliefs converge to the same hypothesis).

The setup we consider in this work is close to the {\em conflicting hypotheses} setup considered in  \cite{nedic2017fast}, where the observations of each agent are generated according to some unknown distribution and the authors show that 
all agents' beliefs converge to the hypothesis that best ``explains" all agents' models.  
In contrast, in this work we are interested in studying the problem where  each agent wants to find its individual true hypothesis, instead of converging to a consensus. 

There are many reasons for which this problem is interesting, as in many cases consensus does not describe the system's behavior.  
Real-life social networks constitute one example, where there are disparate opinions among the various interacting parties. Another example is the scenario where a network of communicating classifiers uses the social learning protocol as in  \cite{bordignon2021network} to classify scenes from different classes. Finally, sensor networks where the agents receive observations generated from different sources is another example of interest.

One main challenge in this problem is the fact that the agents are unaware of which other agents want to find the same state with them. As a result, agents do not know which agents they should cooperate with and whose agents' shared information they should disregard. To tackle the problem, we use the idea that agents' cooperative beliefs are driven by agents' {\em private information}. More specifically, every agent use its private information (i.e., own observations) to form a {\em local belief} about its true state and exchanges this local belief with its neighbors. Based on these exchanged local beliefs the combination weights are formed in a way that is proportional to the probability that the agents want to find the same state. Our contributions are the following. 
\begin{enumerate}
    \item We propose a scheme with adaptive combination weights that utilizes the agents' private information and helps agents in identifying other agents that aim at finding the same hypotheses. In this way we extend the social learning algorithm proposed in \cite{lalitha2014social, lalitha2018social} for the problem of multiple true hypotheses and self-interested agents.
    \item We analyze the asymptotic behavior of agents' beliefs and characterize the agents' belief evolution at steady-state.
    \item We provide sufficient conditions under which the agents in the network manage to learn their true hypotheses.
\end{enumerate}

The problem we study is close to the problem of {\em multi-task learning over networks} studied in \cite{chen2015adaptive,chen2015diffusion,plata2017heterogeneous,zhao2015distributed,nassif2020multitask}. In \cite{chen2015adaptive, chen2015diffusion}, every agent aims at estimating its true parameter vector, which might be different from the target vector of other agents. The authors devise an adaptive combination policy to correctly identify the neighbors with which agents should cooperate to correctly estimate their true parameter vectors. The agents adapt their combination weights based on a Mean Square Deviation (MSD) criterion and a diffusion Least Mean Squares (LMS) algorithm is developed. A different approach was followed in \cite{zhao2015distributed}, where every agent keeps a stand-alone LMS estimate (updated based only on the agent's own signals and not on information from neighbors). At every time instant, every agent performs a binary hypothesis test to decide whether each of its neighbors is interested in the same parameter vector. Related formulations followed in \cite{plata2017heterogeneous,nassif2020multitask,teklehaymanot2015robust,bertrand2010distributed,bertrand2010distributed2,plata2015distributed,bogdanovic2014distributed} and references therein.

The aforementioned works focus on parameter estimation tasks with agents either being aware that they aim at identifying different parameter vectors or not. Here, we focus on the distributed hypothesis testing problem where every agent aims at identifying an underlying hypothesis of interest and is not aware which agents aim at finding the same hypothesis with it or a different one. Thus, our work is  closer to the multi-task decision problem studied in  \cite{marano2021decision}. In \cite{marano2021decision}, an LMS-type algorithm is devised. In contrast, here we study the social learning problem where every agent performs local Bayesian updates before exchanging information with its neighbors. 
Thus, our results neither imply, nor are implied by the results in \cite{marano2021decision}. An interesting result that comes out from our analysis  is the fact that {\em identifiability} (i.e., the ability of an agent to correctly distinguish among the different hypotheses) plays a crucial role on the outcome of the learning process over the network.
\subsection{Notation}
We use the notation $\overset{a.s.}\longrightarrow$ and $\overset{P.}\longrightarrow$ to denote almost sure convergence and convergence in probability, respectively. $\mathbb{I}_{s}$ denotes the indicator function which is equal to $1$ if the statement $s$ is true and $0$ otherwise. $\text{blockdiag}\{A_1,\ldots,A_n\}$ denotes the block-diagonal matrix composed of the matrices $A_1,\ldots,A_n$ and $\mathds{1}$ denotes the all-ones vector. $|\cdot|$ denotes the cardinality of a set.
\section{Problem Formulation}
We assume a set $\mathcal{N}=\{1,\ldots,N\}$ of agents interacting over a network, which is represented by an undirected graph $\mathcal{G}=\langle \mathcal{N}, \mathcal{E} \rangle$, where $\mathcal{E}$ includes bidirectional links between agents. The set of neighbors of an agent $k$ (including agent $k$) is denoted by  $\mathcal{N}_k$. In contrast to the usual setup, here we assume a heterogeneous setting, where there exist multiple true hypotheses that agents want to retrieve. The set of all possible hypotheses is denoted by $\Theta=\{\theta_1,\ldots,\theta_M\}$.

We assume that each agent $k$ has access to observations $\boldsymbol{\zeta}_{k,i}\in\mathcal{Z}_k$ at every time $i\geq1$. Agent $k$ also has access to the likelihood functions $L_k(\zeta_{k,i}|\theta)$, $\theta\in\Theta$. The signals $\boldsymbol{\zeta}_{k,i}$ are independent and identically distributed (i.i.d.) over time. In this work, the sets $\mathcal{Z}_k$ are assumed to be finite. 
We will use the notation $L_k(\theta)$ instead of $L_k(\boldsymbol{\zeta}_{k,i}\vert\theta)$ whenever it is clear from the context. Every agent $k$'s true hypothesis $\theta^{(k)}$ is drawn according to some probability $\mathbb{P}(\boldsymbol{\theta}^{(k)})$ 
initially and remains unchanged throughout the process. Agent $k$'s observations are generated according to the model 
\begin{align}
\boldsymbol{\zeta}_{k,i}\sim L_k(\boldsymbol{\zeta}_{k,i}\vert\boldsymbol{\theta}^{(k)}=\theta^{(k)}), \quad\boldsymbol{\theta}^{(k)}\in\Theta \end{align}
and the states $\boldsymbol{\theta}^{(k)}$ are independent across agents, meaning that $\mathbb{P}(\boldsymbol{\theta}^{(k)},\boldsymbol{\theta}^{(\ell)})=\mathbb{P}(\boldsymbol{\theta}^{(k)})\mathbb{P}(\boldsymbol{\theta}^{(\ell)})$.

Agents' observations are possibly generated by different hypotheses and each agent $k$ aims at finding the realization $\theta^{(k)}$ of its true hypothesis $\boldsymbol{\theta}^{(k)}\in\Theta$ according to which $\boldsymbol{\zeta}_{k,i}$ are created.

Agents share information with their neighbors in a distributed fashion. This information can be utilized to find the underlying true hypothesis by forming {\em beliefs}, which are probability distributions over the set of hypothesis $\Theta$. We consider the log-linear social learning rule \cite{lalitha2014social,bordignon2020social} where the agents update their beliefs, denoted by $\boldsymbol{\nu}_{k,i}$, in the following manner:
\begin{align}
\label{adapt}
    &\boldsymbol{\varphi}_{k,i}(\theta)=
    \frac{L_k(\boldsymbol{\zeta}_{k,i}|\theta)\boldsymbol{\nu}_{k,i-1}(\theta)}{\sum_{\theta'}L_k(\boldsymbol{\zeta}_{k,i}|\theta')\boldsymbol{\nu}_{k,i-1}(\theta')}
    \\
 \label{combine}
    &\boldsymbol{\nu}_{k,i}(\theta)=\frac{\prod_{\ell\in\mathcal{N}_k}(\boldsymbol{\varphi}_{\ell,i}(\theta))^{a_{\ell k}}}{\sum_{\theta'}\prod_{\ell\in\mathcal{N}_k}(\boldsymbol{\varphi}_{\ell,i}(\theta'))^{a_{\ell k}}}, \quad k\in\mathcal{N}
\end{align}
where $a_{\ell k}$ denotes the static (time-invariant) {\em combination weight} assigned by agent $k$ to neighboring agent $\ell$, satisfying $0<a_{\ell k}\leq1$, for all $\ell\in\mathcal{N}_k$, $a_{\ell k}=0$ for all $\ell\notin\mathcal{N}_k$ and $\sum_{\ell\in\mathcal{N}_k}a_{\ell k}=1$. Let $A$ denote the {\em combination matrix}, which consists of all combination weights with $[A]_{\ell k}=a_{\ell k}$. Clearly, $A$ is left-stochastic. 

It is known that if agents use the above algorithm, under the assumption of a strongly connected network (information flows from every agent to any other agent in the network and at least one agent has a self-loop, $a_{kk}>0$) \cite{Sayed_2014}, then the network achieves {\em consensus} \cite{lalitha2014social,nedic2017fast,zhao2012learning,lalitha2018social}, thus ruling out the possibility for agents with different true states to correctly identify them.

In order for agents to find their true state, they should evaluate over time if the received information from the neighbourhood is beneficial to them or not. This means that they have to decide if the information received from their neighbors should be taken into account in the information aggregation step \eqref{combine}. One way to do so is to dynamically adjust the combination weights according to whether agents believe each neighbors  aim at finding the same state or not.
\section{
Adaptive combination weights}
The idea is that the weights assigned by agent $k$ should be zero towards every neighbor that tries to find a different state than $\theta^{(k)}$. However, this information is not known beforehand. 
If an agent can identify its true state alone, then it might be better for that agent not to cooperate and just perform stand-alone Bayesian learning. In this way, it will be guaranteed to converge to its true hypothesis without being misled by other agents. However, some agents might not be able to find their true states alone. This happens when for an agent $k$, its true state $\theta^{(k)}$ is {\em observationally equivalent} to some other $\theta\neq\theta^{(k)}$. In that case this agent will be unable to find its true state without other agents' help. We define the set of states that are observationally equivalent to $\theta^{(k)}$ as follows. 
\begin{Def}\label{obs_equivalence}(\textbf{Observationally equivalent states}). The set
\begin{small}
\begin{align}
    \Theta^{\star}_k\triangleq\left\{\theta\in\Theta:L_k(\zeta_k\vert\theta)=L_k(\zeta_k\vert\theta^{(k)}),\quad\forall\zeta_k\in\mathcal{Z}_k\right\}
\end{align}
\end{small}%
is comprised of all states that are observationally equivalent to $\theta^{(k)}$ for an agent $k\in\mathcal{N}$.\qedsymb
\end{Def}
Note that $\theta^{(k)}$ is always contained in $\Theta^{\star}_k$. 
Before we introduce the adaptive combination weights mechanism, we provide a motivating example. In the network example presented in Fig. \ref{net_example} the set of possible hypotheses is $\Theta=\{\theta_1,\theta_2,\theta_3\}$ and the true hypotheses of agents $1,2,3$ are $\theta_2,\theta_2,\theta_3$, respectively. However, agent $1$ cannot distinguish between hypotheses $\theta_1$ and $\theta_2$. Since agent $1$ communicates with both agents $2$ and $3$, it may not converge to hypothesis $\theta_2$. 
However, if agent $1$ over time realizes that agent $3$'s true hypothesis is $\theta_3$ (i.e., it is different from agent $1$'s true hypothesis), then it can cut off the link with agent $3$ and find its true hypothesis with the aid of agent $2$ (which can find $\theta_2$ alone as $\Theta^{\star}_2=\{\theta_2\}$), provided that agent $2$ also realizes that its true hypothesis is different from agent $3$'s true hypothesis and cuts off its link to agent $3$ as well. Our goal is to devise an adaptive mechanism that enables agents to discriminate over time which agents aim at finding the same hypothesis with them against other agents.
\begin{figure}[!h]
\centering
\includegraphics[width=0.3\textwidth]{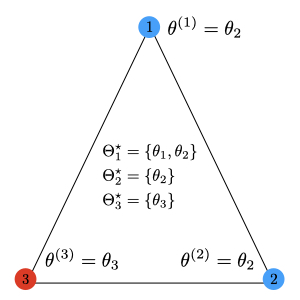}
\caption{A network example with three agents. The true state of agents $1,2$ is $\theta_2$, while the true state of agent $3$ is $\theta_3$.
\label{net_example}}
\end{figure}

In doing so, each agent at every time $i$ can form a {\em local belief} about the unknown hypothesis $\boldsymbol{\theta}^{(k)}$ based only on its own observations $\boldsymbol{\zeta}_{k,1:i}=(\boldsymbol{\zeta}_{k,1},\ldots,\boldsymbol{\zeta}_{k,i})$ until time $i$. These local beliefs do not contain any misleading information from other agents and they are given by
\begin{align}
    &\boldsymbol{\pi}_{k,i}(\theta)=\mathbb{P}(\boldsymbol{\theta}^{(k)}=\theta\vert\boldsymbol{\zeta}_{k,1:i}),\quad\theta\in\Theta
\end{align}
where $\boldsymbol{\pi}_{k,i}$ is the posterior belief over $\boldsymbol{\theta}^{(k)}$ given the sequence of private observations of agent $k$. The belief  $\boldsymbol{\pi}_{k,i}$ can be computed given $\boldsymbol{\pi}_{k,i-1}$ and $\boldsymbol{\zeta}_{k,i}$ recursively according to Bayes' rule:
\begin{align}
\label{HMM_1}
    \boldsymbol{\pi}_{k,i}(\theta)=\frac{L_k(\boldsymbol{\zeta}_{k,i}\vert\theta)\boldsymbol{\pi}_{k,i-1}(\theta)}{\sum_{\theta'\in\Theta}L_k(\boldsymbol{\zeta}_{k,i}\vert\theta')\boldsymbol{\pi}_{k,i-1}(\theta')}.
\end{align}
Now, we can design a scheme based on the local beliefs so that the weights assigned to every neighbor $\ell$ by agent $k$ evolve according to the probability that the two agents are trying to find the same hypothesis (i.e., $\theta^{(k)}=\theta^{(\ell)}$). Let us denote the event that the two agents have the same hypothesis by
    \begin{align}
    \displaystyle\mathcal{S}_{k\ell}\triangleq\{\theta^{(k)}=\theta^{(\ell)}\}=\underset{\theta\in\Theta}\cup\mathcal{S}_{k\ell}^{\theta},\quad k\neq\ell
        \end{align}
        where
        \begin{align}
        \mathcal{S}_{k\ell}^{\theta}\triangleq\{\boldsymbol{\theta}^{(k)}=\theta\cap\boldsymbol{\theta}^{(\ell)}=\theta\},\quad k\neq\ell
        \end{align}
        is the event that both agent $k$ and agent $\ell$'s true state is $\theta$: Since the $\mathcal{S}_{k\ell}^{\theta}$ are disjoint events for different $\theta$
        \begin{align}
        \label{sum_theta}
         \mathbb{P}(\underset{\theta\in\Theta}\cup\mathcal{S}_{k\ell}^{\theta})=\sum_{\theta\in\Theta}\mathbb{P}(\mathcal{S}_{k\ell}^{\theta})
        \end{align}
Obviously, the probability that agent $k$ has the same state with itself is $1$. 
Then, the weight each agent $k$ may assign to its neighbor $\ell$ can be set to
    \begin{small}%
    \begin{align}
    \label{weights_definition}
                \boldsymbol{a}_{\ell k,i}=\begin{cases}\frac{\mathbb{P}(\mathcal{S}_{k\ell}\vert\ \boldsymbol{\zeta}_{k,1:i},\boldsymbol{\zeta}_{\ell,1:i})}{\boldsymbol{\boldsymbol{\sigma}}_{k,i}},&\text{if }\ell\in\mathcal{N}^{\star}_k\\
     \frac{1}{\boldsymbol{\sigma}_{k,i}},&\text{if }\ell=k\\
        0,&\text{otherwise}
        \end{cases}
    \end{align}
        \end{small}%
    where 
    $\mathcal{N}^{\star}_k\triangleq\mathcal{N}_k\setminus\{k\}$ denotes the set of neighbors of agent $k$ without including $k$ and  $\boldsymbol{\sigma}_{k,i}
    $ is a normalizing factor to ensure that $\boldsymbol{A}_i$ is left-stochastic.
    
Construction \eqref{weights_definition} ensures that agent \( k \) incorporates information from agent \( \ell \) in a manner that is proportional to the probability that agents \( k \) and \( \ell \) are observing the same state. As agents gain confidence in their true state over time, this allows them to exclude inconsistent information, and collaborate only with agents who observe data that do not conflict with their local models. 

In the sequel, we first show that agents are able to efficiently compute \( \mathbb{P}(\mathcal{S}_{k \ell}\vert\boldsymbol{\zeta}_{k,1:i},\boldsymbol{\zeta}_{\ell,1:i}) \) and then establish formally that the resulting learning process is consistent.
\begin{Lem} \label{lem1}{\bf (Conditional probability of two agents sharing the same hypothesis)}. 
The probability of two  distinct agents $k,\ell$ having the same state conditioned on the joint observations $\boldsymbol{\zeta}_{k,1:i},\boldsymbol{\zeta}_{\ell,1:i}$ is given by
\begin{align}
\label{probability_same_model_neighbors}
    \mathbb{P}(\mathcal{S}_{k\ell}\vert\boldsymbol{\zeta}_{k,1:i},\boldsymbol{\zeta}_{\ell,1:i})=\sum_{\theta}\boldsymbol{\pi}_{k,i}(\theta)\boldsymbol{\pi}_{\ell,i}(\theta)
\end{align}
\end{Lem}
\begin{proof}
{\em See} Appendix \ref{a1}.
\end{proof}
Utilizing Lemma \ref{lem1}, the normalizing factor is given by
\begin{align}
    \label{normalizing_factor}
        &\boldsymbol{\sigma}_{k,i}
        =1+\sum_{\ell\in\mathcal{N}^{\star}_k}\sum_{\theta\in\Theta}\boldsymbol{\pi}_{k,i}(\theta)\boldsymbol{\pi}_{\ell,i}(\theta).
    \end{align}
Note that according to construction \eqref{weights_definition}, we have that $\boldsymbol{a}_{kk,i}>0$ for all $i\geq1$ and for all $k\in\mathcal{N}$.  In order to account for the information from local beliefs, agents perform two parallel updates. A non-cooperative update, in which the local belief $\boldsymbol{\pi}_{k,i}$ is formed by using \eqref{HMM_1}, which is then shared with every neighbor of $k$; and a social learning update.  
The novel part introduced in the social learning algorithm \eqref{adapt}, \eqref{combine} is that the combination step utilizes the adaptive combination weights $\boldsymbol{A}_i$ instead of static weights. More specifically, every agent $k\in\mathcal{N}$ updates its cooperative belief $\boldsymbol{\mu}_{k,i}$ according to the following procedure:
\begin{align}
\label{adapt_mod}
&\boldsymbol{\psi}_{k,i}(\theta)=
    \frac{L_k(\boldsymbol{\zeta}_{k,i}|\theta)\boldsymbol{\mu}_{k,i-1}(\theta)}{\sum_{\theta'}L_k(\boldsymbol{\zeta}_{k,i}|\theta')\boldsymbol{\mu}_{k,i-1}(\theta')},\quad k\in\mathcal{N}\\
\label{combine_mod}
    &\boldsymbol{\mu}_{k,i}(\theta)=\frac{\prod_{\ell\in\mathcal{N}_k}\boldsymbol{\psi}^{\boldsymbol{a}_{\ell k,i}}_{\ell,i}(\theta)}{\sum_{\theta'}\prod_{\ell\in\mathcal{N}_k}\boldsymbol{\psi}^{\boldsymbol{a}_{\ell k,i}}_{\ell,i}(\theta')}, \quad k\in\mathcal{N}.
\end{align}
We call $\boldsymbol{\mu}_{k,i}$ {\em cooperative beliefs}. For simplicity, and since agents do not have any prior evidence on their true state, we impose the following assumption on the prior local beliefs $\pi_{k,0}(\theta)$ and prior cooperative beliefs $\mu_{k,0}(\theta)$.
\begin{Asm}\label{prior_beliefs}(\textbf{Uniform prior beliefs}). The prior beliefs of all agents are uniform
\begin{align}
    \pi_{k,0}(\theta)=\mu_{k,0}(\theta)=1/|\Theta|,\quad k\in\mathcal{N}, \theta\in\Theta.
\end{align}\qedsymb
\end{Asm}
\subsection{Behavior of Adaptive Weights}
The behavior of the adaptive weights depends on the evolution of local beliefs. The next result characterizes the behavior of local beliefs $\boldsymbol{\pi}_{k,i}$ over time. Before presenting the result, we impose the following technical assumption \cite{nedic2015nonasymptotic}.
\begin{Asm}
\label{non_emptysupport}
{\bf(Likelihood functions with full support)}. $L_k(\zeta\vert\theta)>\alpha$ for some $\alpha>0$ for all $\zeta\in\mathcal{Z}_k$ and for all $\theta\in\Theta$.\qedsymb
\end{Asm}
From \eqref{probability_same_model_neighbors}, it follows that the ability of agent \( k \) to correctly reject inconsistent information from its neighbor \( \ell \) is driven by its ability to reject inconsistent states \( \theta \notin \Theta_k^{\star} \) through \( \boldsymbol{\pi}_{k, i}(\theta) \). We begin by studying its evolution.
\begin{Prp}
\label{rate_localbeliefs} {\bf (Rate of rejection of false hypotheses for local beliefs)}. Under Assumptions \ref{prior_beliefs} and \ref{non_emptysupport}, for all $\theta\notin\Theta^{\star}_k$ the following is true:
\begin{align}
&\mathbb{P}\left(\boldsymbol{\pi}_{k,i}(\theta)\geq\exp\left(x_ki\right)\right)\leq\exp\left(-y_ki\right)
\end{align}
where 
\begin{align}
    &x_k\triangleq-\frac{1}{2}\min_{\theta\notin\Theta^{\star}_k}d_k(\theta)\\
    &y_k\triangleq-\frac{\min_{\theta\notin\Theta^{\star}_k}d^2_k(\theta)}{8(\log\alpha)^2}
\end{align}
and $\alpha$ is given by Assumption \ref{non_emptysupport} and
\begin{align}
    d_k(\theta)\triangleq D_{KL}\Big(L_k(
\theta^{(k)})||L_k(
\theta)\Big)
\end{align}
denotes the KL divergence for an agent $k\in\mathcal{N}$ between $L_k(
\theta^{(k)})$ and $L_k(
\theta)$.
\end{Prp}
\begin{proof}
{\em See} Appendix \ref{a2}.
\end{proof}

Based on the evolution of the local beliefs \( \boldsymbol{\pi}_{k, i}(\theta) \), we can now 
investigate the behavior of the adaptive weights. 
The following result characterizes the asymptotic behavior of the adaptive combination weights.
\begin{Thm} \label{limiting_weights}{\bf (Limiting behavior of the adaptive combination weights)}. 
The adaptive combination weights exhibit the following limiting behavior as $i\to\infty$ for every agent $k\in\mathcal{N}$:
    \begin{align}
    \label{limiting_weightseq}
    \boldsymbol{a}_{\ell k,i}\overset{\text{a.s.}}\longrightarrow\begin{cases}\frac{\eta_{k\ell}}{1+\sum_{\ell'\in\mathcal{N}^{\star}_k}\eta_{k\ell'}},&\text{if }\ell\in\mathcal{N}^{\star}_k\\
            1-\sum_{\ell''\in\mathcal{N}^{\star}_k}\frac{\eta_{k\ell''}}{1+\sum_{\ell'\in\mathcal{N}^{\star}_k}\eta_{k\ell'}},&\text{if }\ell=k\\
        0,&\text{otherwise}
        \end{cases}
\end{align}
where $\eta_{k\ell}\triangleq\frac{|\Theta^{\star}_k\cap\Theta^{\star}_{\ell}|}{|\Theta^{\star}_k||\Theta^{\star}_{\ell}|}$.
\end{Thm}
\begin{proof}{\em See} Appendix \ref{a4}.
\end{proof}
We observe that if an agent $k$ can identify its true hypothesis alone (i.e., $\Theta^{\star}_k=\{\theta^{(k)}\}$), then it will assign asymptotically positive weights only to the neighbors $\ell\in\mathcal{N}^{\star}_k$ for which $\theta^{(k)}$ is within their optimal hypothesis set (i.e., $\theta^{(k)}\in\Theta^{\star}_{\ell}$). We see here the implications of the identifiability capabilities of the agents. For example, if all agents can identify their true hypothesis alone (i.e., $\Theta^{\star}_{k}=\{\theta^{(k)}\}$ for all $k\in\mathcal{N}$), then the network will (asymptotically) decompose into components where every agent communicates only with the neighbors that aim at finding the same hypothesis with it.
\section{Analysis of the algorithm}
In this section we examine whether 
the adaptive combination scheme is sufficient to drive the agents' cooperative beliefs $\boldsymbol{\mu}_{k,i}$ to the individually correct hypotheses. First, we observe from Theorem \ref{limiting_weights} that the combination matrix $\boldsymbol{A}_i$ converges to a limiting matrix $A_{\infty}$ with elements $[A_{\infty}]_{\ell k}=a_{\ell k,\infty}$, defined as
\begin{align}
    A_{\infty}\triangleq\lim_{i\to\infty}\boldsymbol{A}_i.
\end{align}
In order to study the evolution of the cooperative beliefs $\boldsymbol{\mu}_{k,i}$ generated by our proposed algorithm with adaptive combination weights \eqref{adapt_mod}-\eqref{combine_mod}, we will show that they track the evolution of beliefs generated by the algorithm \eqref{adapt}-\eqref{combine} with the steady-state combination matrix $A_{\infty}$, which is much simpler to analyze. More specifically, we will show that asymptotically $\boldsymbol{\mu}_{k,i}$ tracks $\boldsymbol{\mu}^c_{k,i}$ which is given by 
\begin{align}
\label{adapt_centralized}
&\boldsymbol{\psi}^c_{k,i}(\theta)=
    \frac{L_k(\boldsymbol{\zeta}_{k,i}|\theta)\boldsymbol{\mu}^c_{k,i-1}(\theta)}{\sum_{\theta'}L_k(\boldsymbol{\zeta}_{k,i}|\theta')\boldsymbol{\mu}^c_{k,i-1}(\theta')},\quad k\in\mathcal{N}
    \\
\label{combine_centralized}
    &\boldsymbol{\mu}^c_{k,i}(\theta)=\frac{\prod_{\ell\in\mathcal{N}_k}(\boldsymbol{\psi}^c_{\ell,i}(\theta))^{a_{\ell k,\infty}}}{\sum_{\theta'}\prod_{\ell\in\mathcal{N}_k}(\boldsymbol{\psi}^c_{\ell,i}(\theta'))^{a_{\ell k,\infty}}},\quad k\in\mathcal{N}.
\end{align}
The evolution of the beliefs generated by \eqref{adapt_centralized}-\eqref{combine_centralized} has been analyzed for a time-invariant combination matrix for both cases of strongly-connected \cite{lalitha2014social,bordignon2020social} and weakly-connected networks \cite{matta2019graph}. First we prove a useful Lemma that characterizes the structure of $A_{\infty}$.
\begin{Lem}\label{A_inf}
\textbf{(Structure of }$A_{\infty}$\textbf{)}. $A_{\infty}$ is comprised of $S\in\mathbb{N}$ disjoint strongly-connected components, meaning
\begin{align}
    A_{\infty}=\begin{pmatrix}
&A_{\infty,1} & \dots &\boldsymbol{0}\\
&\vdots &\ddots &\vdots \\
&\boldsymbol{0} & \dots & A_{\infty,S}
\end{pmatrix}.
    \end{align}
Then, we have
\begin{align}
\label{Per_eq}
    \bar{A}^{\mathsf{T}}_{\infty}\triangleq\lim_{t\to\infty}(A^{\mathsf{T}}_{\infty})^t=\text{blockdiag}\{p_1\mathds{1}^{\mathsf{T}},\ldots,p_S\mathds{1}^{\mathsf{T}}\}.
\end{align}
$p_s$, $s\in\{1,\ldots,S\}$ is the Perron eigenvector of $A_{\infty,s}$.
\end{Lem}
\begin{proof}
{\em See} Appendix \ref{a8}.
\end{proof}
Let us define the set $\bar{\mathcal{N}}_s$, $s\in\{1,\ldots,S\}$ as the set of agents whose combination weights comprise $A_{\infty,s}$. Furthermore, let us define for $\bar{\mathcal{N}}_s$ the {\em sub-network confidence} for a state $\theta\in\Theta$ as
\begin{align}
\label{confidence}
    C_s(\theta)\triangleq-\sum_{k\in\bar{\mathcal{N}}_s}p_s(k)D_{KL}(L_k(\theta^{(k)})||L_k(\theta)).
\end{align}
where $p_s(k)$ is the $k^{th}$ element of $p_s$ and let
\begin{align}
\label{confidence_max}
    \bar{\Theta}^{\star}_s\triangleq\left\{\theta^{\star}_s\triangleq\arg\max_{\theta\in\Theta}C_s(\theta)\right\}.
\end{align}
This set is comprised of the hypotheses that best describe the sub-network agents' observation models weighted by their centrality. Now, we can provide the main result, which characterizes the evolution of cooperative beliefs $\boldsymbol{\mu}_{k,i}$. 
\begin{Thm}\label{public_b_c_t}({\bf  Cooperative beliefs convergence and consistent learning}). For any agent $k\in\bar{\mathcal{N}}_s$:\vspace{-2.5mm}\begin{enumerate}
    \item The cooperative beliefs converge to $0$ in probability, meaning:
$\boldsymbol{\mu}_{k,i}(\theta)\overset{P.}\longrightarrow0$, 
for every $\theta\notin\bar{\Theta}^{\star}_s$.
\item Agent $k$ learns its true state, meaning 
    $\boldsymbol{\mu}_{k,i}(\theta^{(k)})\overset{P.}\longrightarrow1$, 
if 
$\bar{\Theta}^{\star}_s=\{\theta^{(k)}\}$.
\end{enumerate}
\end{Thm}\vspace{-2.5mm}
\begin{proof}{\em See} Appendix \ref{a7}.
\end{proof}
There is one more question of interest to answer. We see from the result above that whether an agent is able to learn its true state is dependent on the structure of the sub-network $\bar{\mathcal{N}}_s$. However, from Theorem  \ref{limiting_weights} we see that this structure depends on the identifiability capabilities of the agents (i.e., on the sets $\Theta^{\star}_k$) and on the graph topology given by $\mathcal{G}$. The following result provides conditions that guarantee that {\em every} agent in the network will find its true state.
\begin{Cor}\label{consistent_learning2}
\textbf{(Globally consistent learning)}. 
Under the proposed adaptive combination scheme, every agent $k\in\mathcal{N}$ learns its true state, meaning
\begin{align}
    \boldsymbol{\mu}_{k,i}(\theta^{(k)})\overset{P.}\longrightarrow1,\quad\forall k\in\mathcal{N}
\end{align}
if both of the following hold:
\begin{align}\label{c1}
&\Theta^{\star}_k\cap\Theta^{\star}_{\ell}=\emptyset,\forall k\in\mathcal{N},\forall \ell\in\mathcal{N}^{\star}_k\text{ such that  }\theta^{(k)}\neq\theta^{(\ell)}\\\label{c2}
&\underset{\ell\in\bar{\mathcal{N}}_s}\cap\Theta^{\star}_{\ell}=\{\theta^{(k)}\},\forall s\in\{1,\ldots,S\}\text{ such that }k\in\bar{\mathcal{N}}_s.
\end{align}
\end{Cor}
\begin{proof}
If \eqref{c1} holds for two neighbors $k,\ell$ with different states, then we have from Theorem \ref{limiting_weights} that $a_{k\ell,\infty}=a_{\ell k,\infty}=0$, which means that two agents with different states do not exchange information directly. Since this holds for every two neighbors across the network, this implies that there can not be two agents with different true states in the same sub-network $\bar{\mathcal{N}}_s$. This further implies that, in every sub-network all agents have the same true hypothesis, because $\theta^{(k)}\in\Theta^{\star}_k$ for all $k\in\mathcal{N}$.

Then, if \eqref{c2} holds, we have that for every $\theta\neq\theta^{(k)}$ there is at least one agent $k\in\bar{\mathcal{N}}_s$ such that  $d_k(\theta)>0$. This implies that $C_s(\theta)<0$ for all $\theta\neq\theta^{(k)}$. Then, from \eqref{confidence} we have that $C_s(\theta^{(k)})=0$, which implies that $\theta^{\star}_s=\theta^{(k)}$. Thus, $\bar{\Theta}^{\star}_s=\{\theta^{(k)}\}$ and part $2$ of Theorem \ref{public_b_c_t} applies.
\end{proof}

Condition \eqref{c1} ensures that for any two neighboring agents $k,\ell$ with different states both agents can rule out the state of the other agent based on their own observations (i.e., $\theta^{(k)}\notin\Theta^{\star}_{\ell}$ and $\theta^{(\ell)}\notin\Theta^{\star}_k$). This condition ensures that in every formed sub-network all agents share the same true state. Then, condition \eqref{c2} is needed to ensure that in every sub-network the agents can collectively identify their true state.  We illustrate the results by means of some numerical examples.
\begin{Exa}{\textbf{(Consistent learning for all agents)}.}
We refer to Figure 1 for which from 
Theorem \ref{limiting_weights} we have
\begin{align}
    A_{\infty}=\begin{pmatrix}
&2/3 &1/3 & 0\\
&1/3 &2/3 &0\\
&0 &0 &1
\end{pmatrix}
\end{align}
where the order of agents' labeling is $\{1,2,3\}$. As we observe, $A_{\infty}$ is comprised of two strongly connected components. The limiting matrix in this case is
\begin{align}
    \bar{A}^{\mathsf{T}}_{\infty}=\begin{pmatrix}
&1/2 & 1/2 &0\\
&1/2 &1/2 &0\\
&0 &0 &1.
\end{pmatrix}
\end{align}
Moreover, because \eqref{c1}, \eqref{c2} hold for all $3$ agents, all of them converge to their true hypothesis, as we can see in the left plot of Fig. \ref{beliefs_correct}.\qedsymb
\begin{figure}[!h]
\centering\hspace{-5mm}
\includegraphics[width=0.5\textwidth]{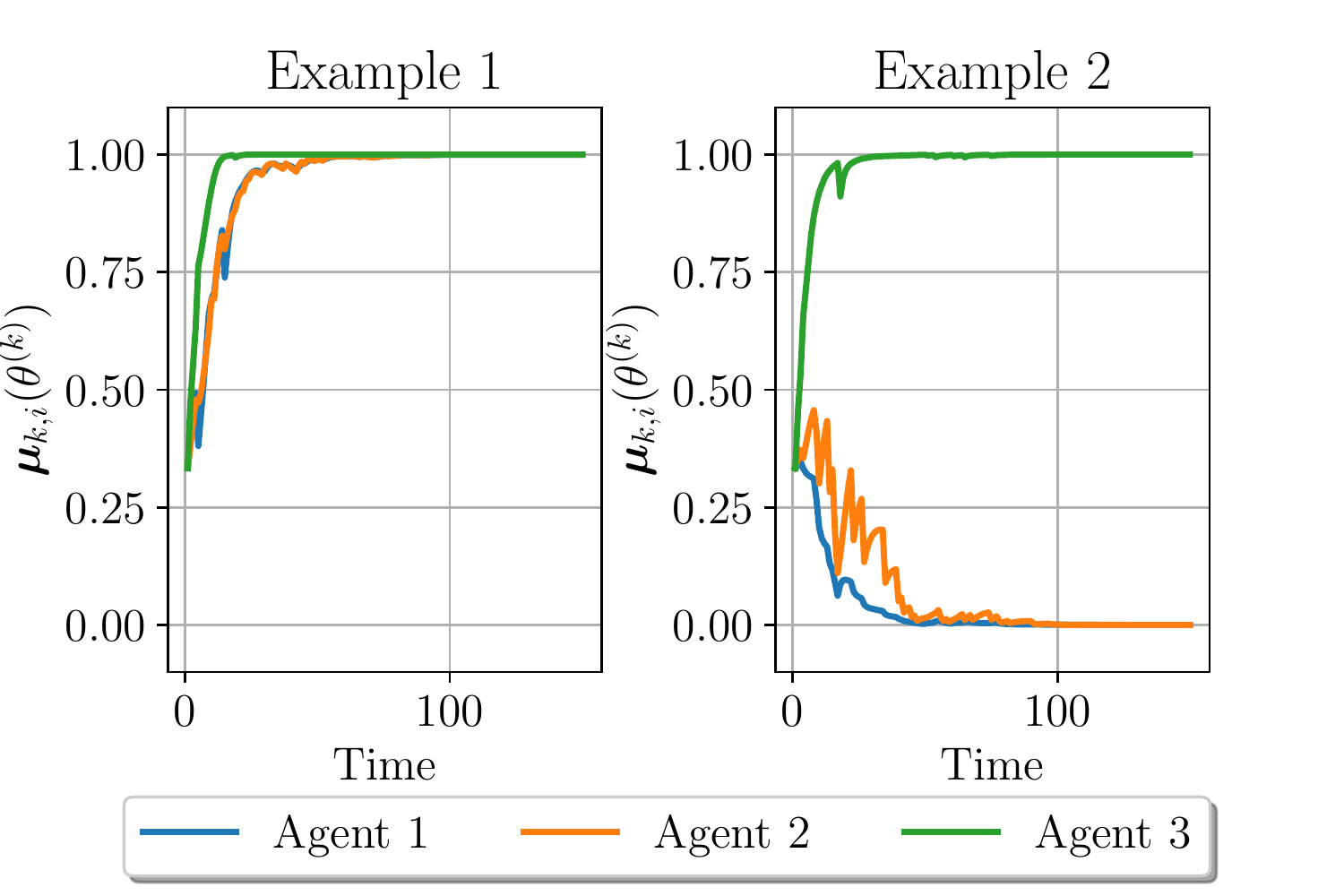}
\caption{Agents' evolution of the beliefs on each agent's true hypothesis (i.e, $\boldsymbol{\mu}_{k,i}(\theta^{(k)}),k\in\{1,2,3\}$) for: ({\em left} Example $1$) and ({\em right}) Example $2$.}
\label{beliefs_correct}
\end{figure}
\end{Exa}
\begin{Exa}{\textbf{(Inconsistent learning})}.
However, in case we have $\Theta^{\star}_1=\{\theta_1,\theta_2,\theta_3\}$ for the same example, then
\begin{align}
    A_{\infty}=\begin{pmatrix}
&4/6 &1/4 &1/4\\
&1/6 &3/4 &0\\
&1/6 &0 &3/4
\end{pmatrix}
\end{align}
and
\begin{align}
    \bar{A}^{\mathsf{T}}_{\infty}\approx\begin{pmatrix}
&0.428 & 0.286 &0.286\\
&0.428 &0.286 &0.286\\
&0.428 &0.286 &0.286
\end{pmatrix}
\end{align}
As we see above  $(A^{\mathsf{T}}_{\infty})^i$ converges to a rank-$1$ matrix. Thus, the network achieves consensus and all agents' beliefs converge to the same hypothesis ($\theta_3$ in this example). As a result, agents $1$ and $2$ fail to identify their true states as we observe in the right plot of Fig. \ref{beliefs_correct}. 
Also, note that in this example \eqref{c1},  \eqref{c2} are violated.\qedsymb
\end{Exa}
As we see from Example $2$, consistent learning is not always possible. When an agent cannot find its true hypothesis alone, then their beliefs will be determined by the information received from its neighborhood.

Another interesting remark is the following. In Example $2$ one agent (agent $1$) cannot rule out any of the three states. This makes information flow from agent $2$ to agent $3$ and their cooperative beliefs achieve consensus. Because of that, agent $2$'s cooperative belief converges to $\theta_3$, despite the fact that agent $2$ can find its true state ($\theta_2$) alone. This case has to be taken into account. Agents whose private beliefs indicate that a particular state $\theta$ is unlikely to be their true hypothesis (i.e., $\boldsymbol{\pi}_{k,i}(\theta)\to0$), then this information should be taken into account and rule out cooperative beliefs that suggest the opposite. One way to do that is by setting a low threshold $\epsilon>0$. If for a state $\theta$, $\boldsymbol{\pi}_{k,i}(\theta)<\epsilon$ and simultaneously $\boldsymbol{\mu}_{k,i}(\theta)>\epsilon$, then the cooperative belief $\boldsymbol{\mu}_{k,i}$ should be disregarded and the agent should limit itself to using its own local belief for inferring its true hypothesis. To do so, we assume that every agent keeps a {\em global} belief vector which is given by
\begin{align}
\label{global_belief}
    \bar{\boldsymbol{\mu}}_{k,i}=\begin{cases}\boldsymbol{\pi}_{k,i},&\text{if }\exists\theta:\boldsymbol{\pi}_{k,i}(\theta)<\epsilon\text{ and }\boldsymbol{\mu}_{k,i}(\theta)>\epsilon,\\
    \boldsymbol{\mu}_{k,i},&\text{otherwise}.
        \end{cases}
\end{align}
Finally, the proposed social learning algorithm is summarized as follows.
\noindent\rule{\linewidth}{0.5mm} \\[-0.5mm]
\textbf{Algorithm. Self Aware Social learning (SASL)}.\\[-2mm]
\rule{\linewidth}{0.5mm}
Initialize $\mu_{k,0}(\theta)=\pi_{k,0}(\theta)$ for all $k\in\mathcal{N},\theta\in\Theta$.\\
For all $k\in\mathcal{N}$ and $i\geq1$:\vspace{-0.25cm}
\begin{enumerate}
\item Obtain $\boldsymbol{\zeta}_{k,i}$.
\item Update for all $\theta\in\Theta$\\ $\boldsymbol{\psi}_{k,i}(\theta)=\frac{L_k(\boldsymbol{\zeta}_{k,i}|\theta)\boldsymbol{\mu}_{k,i-1}(\theta)}{\sum_{\theta'}L_k(\boldsymbol{\zeta}_{k,i}|\theta')\boldsymbol{\mu}_{k,i-1}(\theta')},\quad k\in\mathcal{N}$.
\item Update for all $\theta\in\Theta$\\
$\boldsymbol{\pi}_{k,i}(\theta)=\frac{L_k(\boldsymbol{\zeta}_{k,i}\vert\theta)\boldsymbol{\pi}_{k,i-1}(\theta)}{\sum_{\theta'\in\Theta}L_k(\boldsymbol{\zeta}_{k,i}\vert\theta')\boldsymbol{\pi}_{k,i-1}(\theta')}$.
\item Exchange $\boldsymbol{\pi}_{k,i}$ with every $\ell\in\mathcal{N}^{\star}_k$.
\item Update 
$\boldsymbol{a}_{\ell k,i}$ via \eqref{weights_definition}.
\item Update for all $\theta\in\Theta$\\ $\boldsymbol{\mu}_{k,i}(\theta)=\frac{\prod_{\ell\in\mathcal{N}_k}\boldsymbol{\psi}^{\boldsymbol{a}_{\ell k,i}}_{\ell,i}(\theta)}{\sum_{\theta'}\prod_{\ell\in\mathcal{N}_k}\boldsymbol{\psi}^{\boldsymbol{a}_{\ell k,i}}_{\ell,i}(\theta')}, \quad k\in\mathcal{N}$.
\vspace{0.1cm}
\item Compute global belief $\bar{\boldsymbol{\mu}}_{k,i}$ via \eqref{global_belief}.
\end{enumerate}
\vspace{-0.35 cm}\rule{\linewidth}{0.5mm}\\[-0.5mm]
\section{Experiments}
In the following experiments we illustrate the agents' belief evolution for a network of $10$ agents. The network is depicted in Fig. \ref{network}.
\begin{figure}[!h]
\centering\hspace{-5mm}
\includegraphics[width=0.5\textwidth]{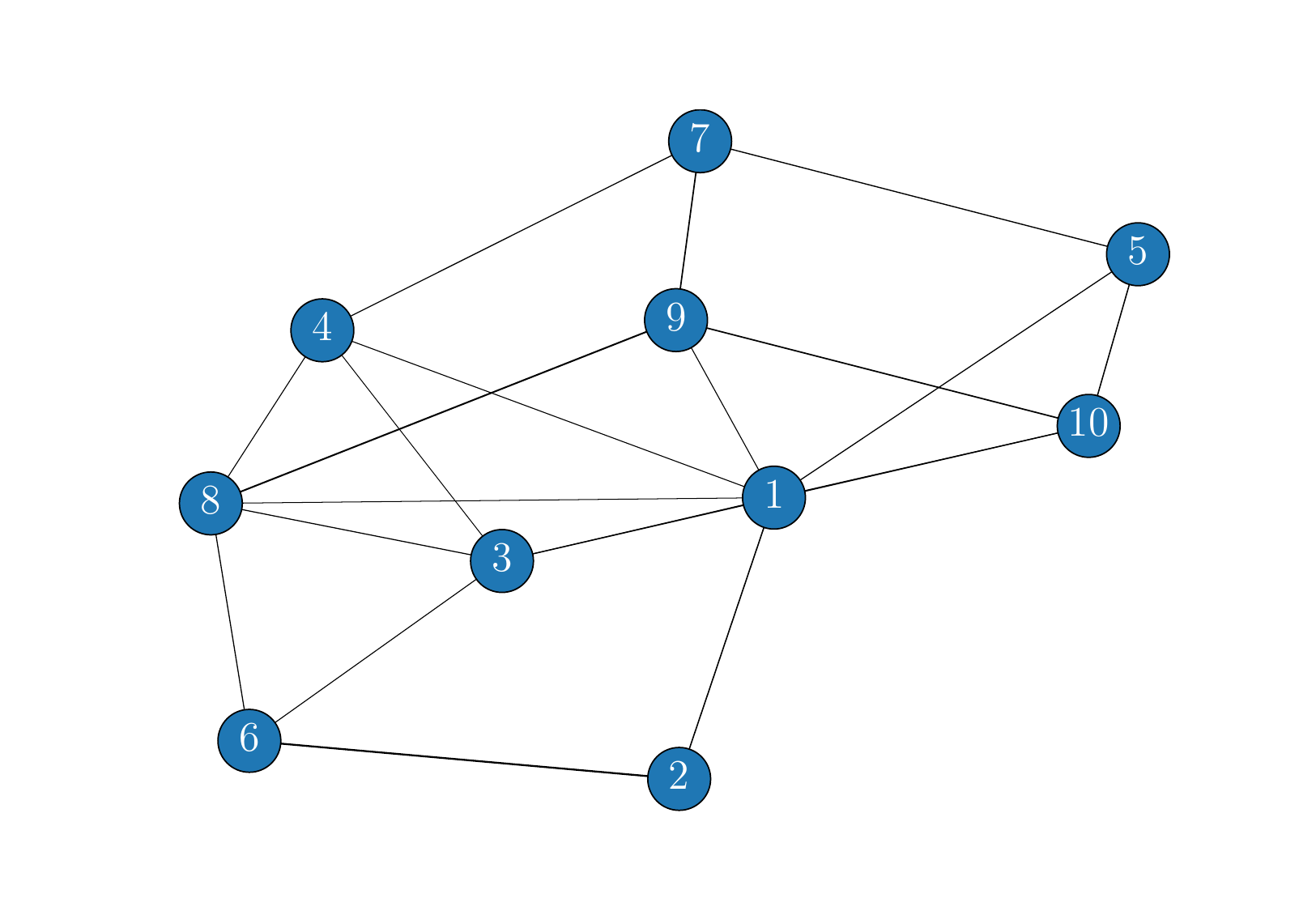}
\caption{Network topology.}
\label{network}
\end{figure}
To facilitate the illustration of our results in a simple way we assume that $|\mathcal{Z}_k|=10$ for all $k\in\mathcal{N}$ and the set of possible hypotheses is $\Theta=\{\theta_1,\ldots,\theta_{10}\}$. In order highlight the need for an adaptive combination mechanism in updating the cooperative beliefs, we compare the asymptotic beliefs of our proposed scheme to the classic cooperative social learning solution \eqref{adapt} - \eqref{combine} with a static (time-invariant) combination matrix $A$ (every agent assigns uniform combination weights to its neighbors) as well as to the non-cooperative learning (beliefs are given by \eqref{HMM_1}).

First, we explore a simple setup where every agent has a distinct true state, meaning $\theta^{(k)}=\theta_k$, for all $k\in\mathcal{N}$. Furthermore, none of the agents faces an identification problem and agents' likelihood functions are given by the following expression:
\begin{align}
\begin{small}
\label{likelihood_functions}
    L_k(\zeta^y_{k,i}\vert\theta_x)=\begin{cases}q_k\in(0,1),&\text{if }y=x,\\
    \frac{1-q_k}{|\mathcal{Z}_k|-1},&\text{otherwise}.
        \end{cases}
\end{small}%
\end{align}
for all $k\in\mathcal{N}$ and for $x,y,k=1,\ldots,10$. $\zeta^y_{k,i}\in\mathcal{Z}_k$ denotes the $y^{th}$ observation of agent $k$. We set $q_k=0.28$ for all $k\in\mathcal{N}$. Note that for the likelihood functions given by \eqref{likelihood_functions} and for this value of $q_k$, we have $\Theta^{\star}_k=\{\theta^{(k)}\}$ for all $k\in\mathcal{N}$. In Fig. \ref{network_3} (third row) we observe that each agent $k$'s cooperative beliefs  converge to its true hypothesis $\theta^{(k)}$ for all $k\in\mathcal{N}$.  We use a light green to orange colormap to indicate the magnitude of agents' beliefs on their true state. Light green indicates beliefs that are close to $0$, while orange indicates beliefs close to $1$. Also, note that the conditions given by Corollary \ref{consistent_learning2} are satisfied. On the contrary, the cooperative social learning algorithm leads all agents, except for agent $1$, to inconsistent learning, as expected due to the fact that the network achieves consensus and agents beliefs converge to $\theta_1$, which is the hypothesis maximizing  \eqref{confidence} in this example. Also note that the the non-cooperative learning (second row) is consistent since every agent can identify its true hypothesis alone. Finally, note that for SASL in steady-state the network decomposes into isolated agents as all adaptive weights assigned from every agent to its neighbors go to $0$ and thus no information is exchanged across the network. This is depicted in third row in Fig. \ref{network_3} where there are no edges as we set the edge-width between any two connected agents to be proportional to the sum of the respective combination weights between the agents. The same rationale is followed in the other experiments as well.
\begin{figure}[!h]
\centering\hspace{-5mm}
\includegraphics[width=0.35\textwidth]{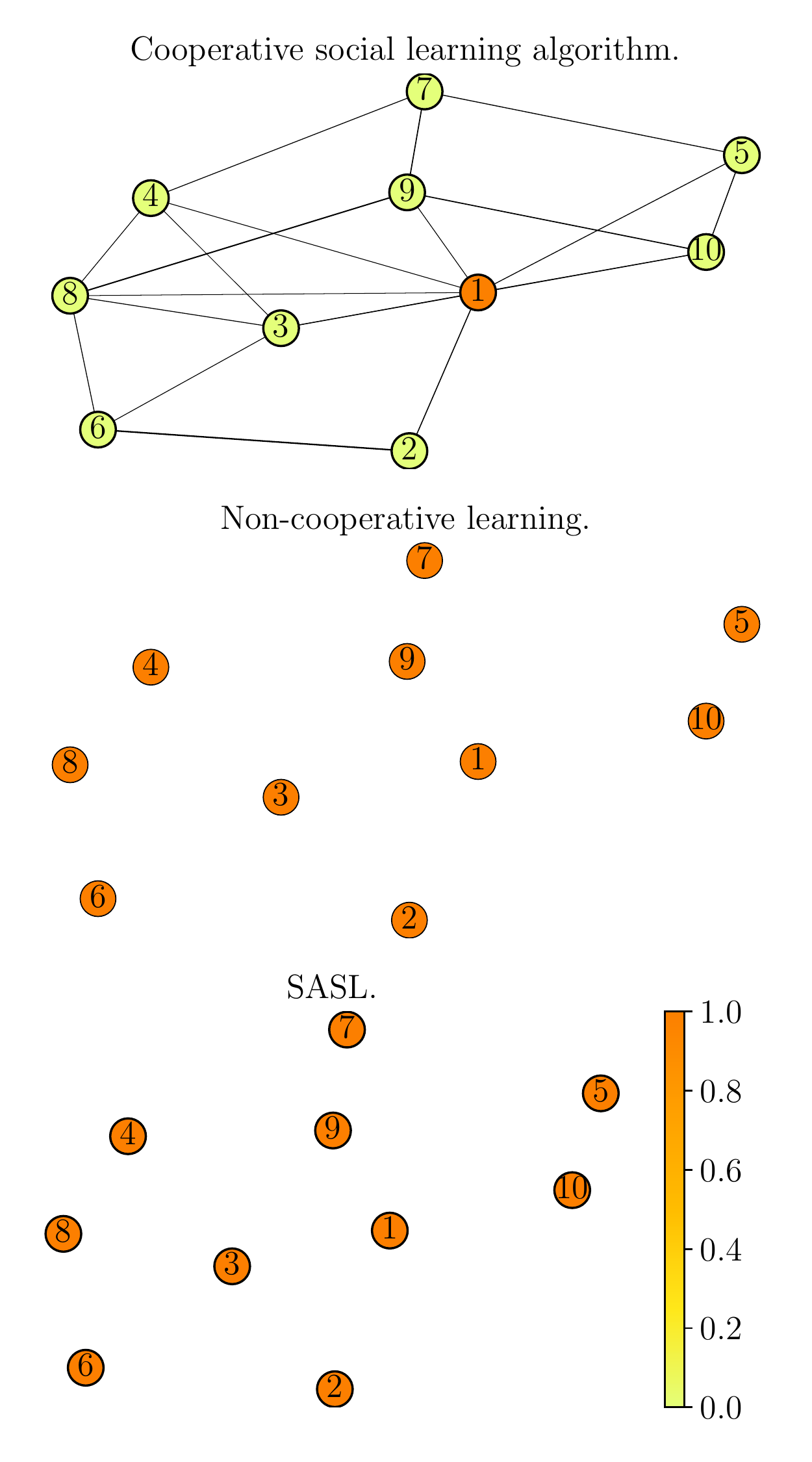}
\caption{
Steady-state beliefs of every agent on its true hypothesis $\theta^{(k)}$ for cooperative social learning (i.e., $\boldsymbol{\nu}_{k,i}(\theta^{(k)})$) (first row), non-cooperative learning (i.e., $\boldsymbol{\pi}_{k,i}(\theta^{(k)})$) (second row) and SASL Algorithm (i.e., $\boldsymbol{\mu}_{k,i}(\theta^{(k)})$) (third row). {\em Colormap} explanation: Agents colored in orange indicate that their beliefs on their true state are close to $1$, while light green denotes that beliefs are close to $0$}.
\label{network_3}
\end{figure}

We consider next a more interesting scenario where some agents share the same true hypothesis and some agents face an identification problem. More specifically, the agents' true hypotheses are assigned as follows:
\begin{align}
\label{true_hypotheseseq}
    \theta^{(k)}=\begin{cases}\theta_1,&\text{if }k\in\{1,\ldots,5\}\\
    \theta_6,&\text{if } k\in\{6,\ldots,10\}.
        \end{cases}
\end{align}
The agents' likelihood functions are constructed as follows so that some agents cannot discriminate among some states. For agents $1$ and $6$, their likelihood functions are given by \eqref{likelihood_functions} (they can identify their true hypotheses alone). For agents $2,3,4,5$  $L_k(\zeta^y_{k,i}\vert\theta_{x})$ is given by \eqref{likelihood_functions} for $x\geq 6$ and by
\begin{align}
\label{likelihood_functions2}
    L_k(\zeta^y_{k,i}\vert\theta_{x})=
    \frac{1}{|\mathcal{Z}_k|},\quad\forall y=1,\ldots,|\mathcal{Z}_k|
\end{align}
for $x\leq 5$. For agents $7,8,9,10$  $L_k(\zeta^y_{k,i}\vert\theta_{x})$ is given by \eqref{likelihood_functions} for $x\leq 5$ and by \eqref{likelihood_functions2} for $x\geq 6$. 

In this case we see that \eqref{c1} is satisfied and the network decomposes into two strongly connected components, one consisting of agents $1,2,3,4,5$ and one consisting of agents $6,7,8,9,10$ ({\em see} third row in Fig. \ref{beliefs_commonc}). We also can verify from \eqref{likelihood_functions2} that condition \eqref{c2} holds. As we see in third row in Fig. \ref{beliefs_commonc} (third row), all agents converge to their true hypotheses, as expected by Corollary \ref{consistent_learning2}, while both cooperative and non-cooperative solutions lead to inconsistent learning for some of the agents.
\begin{figure}[!h]
\centering\hspace{-5mm}
\includegraphics[width=0.35\textwidth]{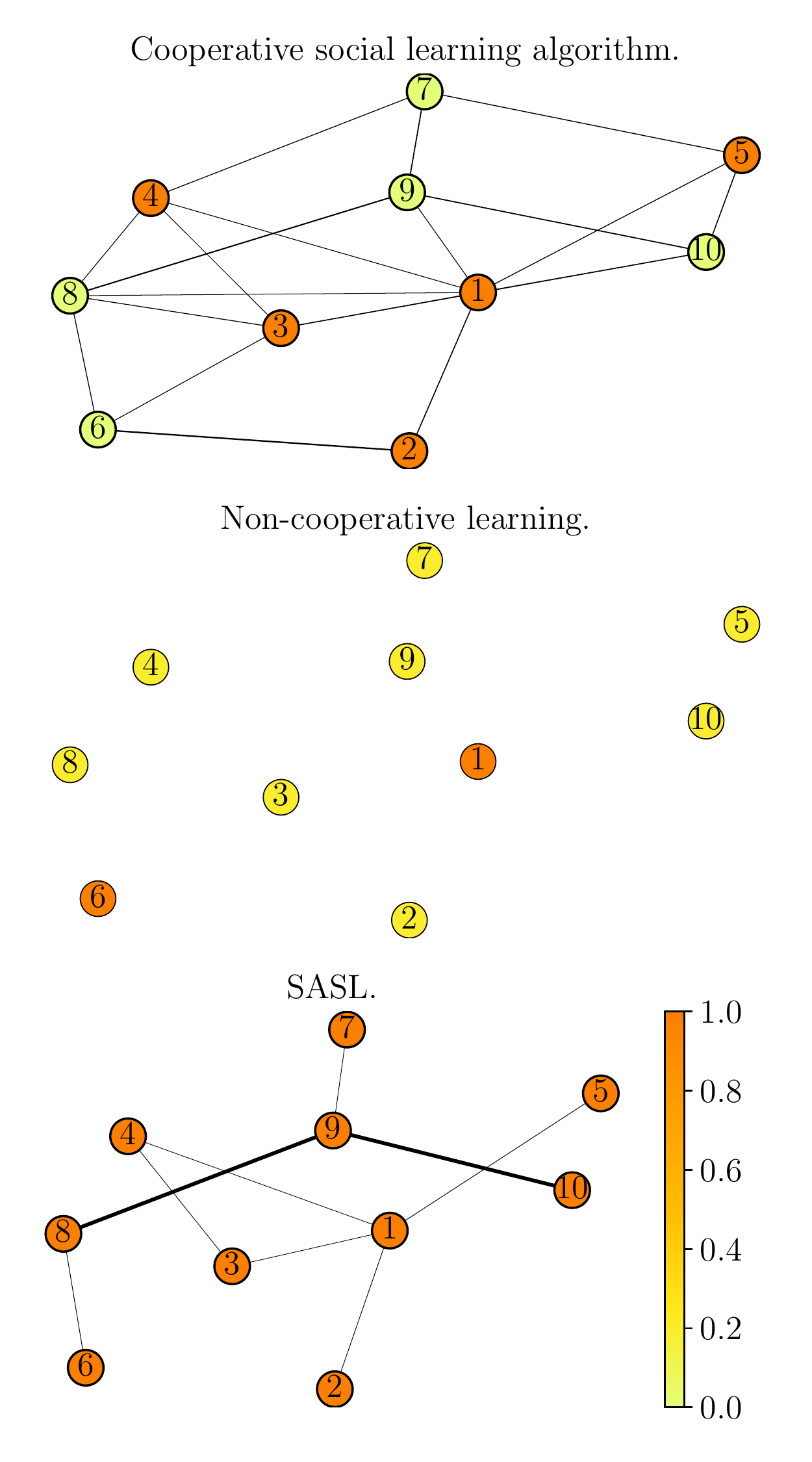}
\caption{Steady-state beliefs of every agent on its true hypothesis $\theta^{(k)}$ for cooperative social learning (i.e., $\boldsymbol{\nu}_{k,i}(\theta^{(k)})$) (first row), non-cooperative learning (i.e., $\boldsymbol{\pi}_{k,i}(\theta^{(k)})$) (second row) and SASL Algorithm (i.e., $\bar{\boldsymbol{\mu}}_{k,i}(\theta^{(k)})$) (third row) for all agents.}
\label{beliefs_commonc}
\end{figure}

Finally, we examine a scenario where the proposed solution fails to achieve consistent learning for all agents. In this setup, some agents again face an identification problem, but conditions \eqref{c1} and \eqref{c2} will not be satisfied. In this setup the true hypotheses are given again by \eqref{true_hypotheseseq}, the likelihood functions of agents $1$ and $6$ are given by \eqref{likelihood_functions2}, but now the likelihood functions for the remaining agents are given by
\begin{align}
    \label{likelihood_functions3}
    L_k(\zeta^y_{k,i}\vert\theta_{x})=\frac{1}{|\Theta|},\quad\text{for all }y, x,\text{ and for }k\neq 1,6.
\end{align}
We see that in this case all agents except for $1,6$ cannot distinguish between any two hypotheses, meaning $\Theta^{\star}_k=\Theta$ for all $k\neq1,6$. Moreover, we can verify
that condition \eqref{c1} does not hold and as a result the network does not decompose into two disconnected components ({\em see} third row in Fig. \ref{beliefs_commonnc}). 
As a result, the network achieves consensus and all agents' beliefs converge to hypothesis $\theta_1$ except for agent $6$ that can identify its true hypothesis alone. Of course, the cooperative and non-cooperative solutions lead to inconsistent learning for some agents as well, with the results being strictly worse compared to SASL (agent $6$ does not converge to $\theta_6$ with the cooperative solution, while under the non-cooperative solution agents $2,3,4,5$ do not converge to their true hypothesis $\theta_1$).
\begin{figure}[!h]
\centering\hspace{-5mm}
\includegraphics[width=0.35\textwidth]{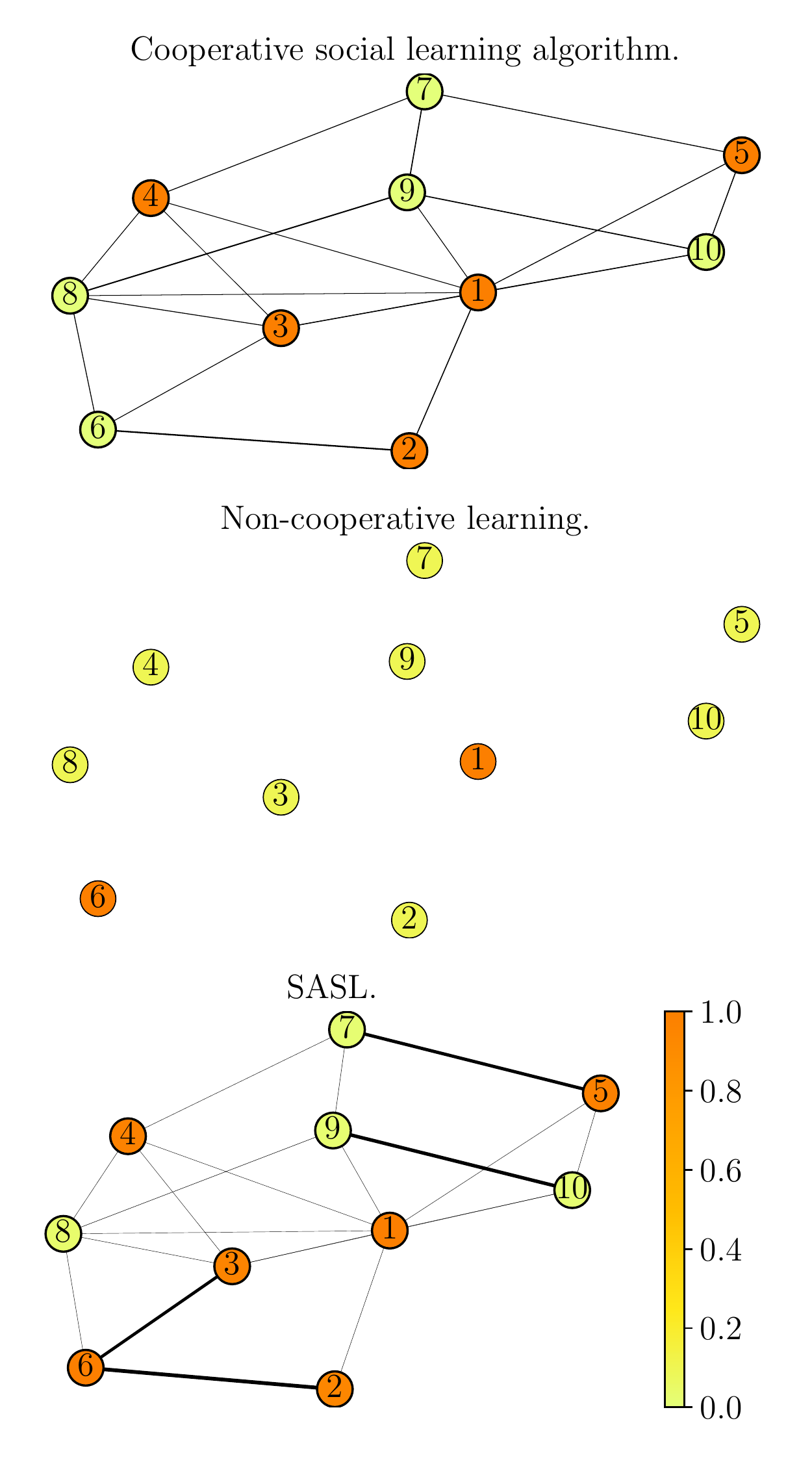}
\caption{Steady-state beliefs of every agent on its true hypothesis $\theta^{(k)}$  for cooperative social learning (i.e., $\boldsymbol{\nu}_{k,i}(\theta^{(k)})$) (first row), non-cooperative learning (i.e., $\boldsymbol{\pi}_{k,i}(\theta^{(k)})$) (second row) and SASL Algorithm (i.e., $\bar{\boldsymbol{\mu}}_{k,i}(\theta^{(k)})$) (third row).}
\label{beliefs_commonnc}
\end{figure}
\section{Conclusions}
In this work the problem of social learning with multiple true hypotheses and self-interested agents was investigated. Contrary to previous works that aim at showing that the network achieves consensus, here we investigated the scenario where every agent wants to converge to its true hypothesis. For this reason, we devised an adaptive combination weights scheme based on agents' private information and studied the performance of the proposed social learning algorithm. We provided conditions under which every agent in the network successfully learns its true hypothesis and we illustrated the learning behavior of the agents via computer simulations.
\appendices
\section{Proof of Lemma  \ref{lem1}}
\label{a1}
Let us define for convenience the following:
\begin{align}
&p_{k,i\vert\theta}\triangleq\mathbb{P}(\boldsymbol{\zeta}_{k,1:i}\vert\boldsymbol{\theta}^{(k)}=\theta),\quad\theta\in\Theta\\
&p_{\ell,i\vert\theta}\triangleq\mathbb{P}(\boldsymbol{\zeta}_{\ell,1:i}\vert\boldsymbol{\theta}^{(\ell)}=\theta),\quad\theta\in\Theta\\
&q_{\bar{\theta}^{(k)}}\triangleq\mathbb{P}(\boldsymbol{\theta}^{(k)}=\bar{\theta}^{(k)}),\quad\bar{\theta}^{(k)}\in\Theta\\
&q_{\bar{\theta}^{(\ell)}}\triangleq\mathbb{P}(\boldsymbol{\theta}^{(\ell)}=\bar{\theta}^{(\ell)}),\quad\bar{\theta}^{(\ell)}\in\Theta.
\end{align}
For every two agents $k\neq\ell$ let us consider the joint conditional probability 
\begin{small}
\begin{align}
\label{products_beliefs}
    &\mathbb{P}(\mathcal{S}_{k\ell}\vert\boldsymbol{\zeta}_{k,1:i},\boldsymbol{\zeta}_{\ell,1:i})\nonumber\\
    &\overset{\eqref{sum_theta}}=\sum_{\theta\in\Theta}\mathbb{P}(\boldsymbol{\theta}^{(k)}=\theta,\boldsymbol{\theta}^{(\ell)}=\theta\vert\boldsymbol{\zeta}_{k,1:i},\boldsymbol{\zeta}_{\ell,1:i})\nonumber\\
    &\overset{(a)}=\sum_{\theta\in\Theta}\frac{p_{k,i\vert\theta}p_{\ell,i\vert\theta}\mathbb{P}(\boldsymbol{\theta}^{(k)}=\theta,\boldsymbol{\theta}^{(\ell)}=\theta)}{\mathbb{P}(\boldsymbol{\zeta}_{k,1:i},\boldsymbol{\zeta}_{\ell,1
    :i})}\nonumber\\
    &\overset{(b)}=\sum_{\theta\in\Theta}\frac{p_{k,i\vert\theta}p_{\ell,i\vert\theta}\mathbb{P}(\boldsymbol{\theta}^{(k)}=\theta)\mathbb{P}(\boldsymbol{\theta}^{(\ell)}=\theta)}{\sum_{\bar{\theta}^{(k)},\bar{\theta}^{(\ell)}}p_{k,i\vert\bar{\theta}^{(k)}}p_{\ell,i\vert\bar{\theta}^{(\ell)}}q_{\bar{\theta}^{(k)}}q_{\bar{\theta}^{(\ell)}}}\nonumber\\
    &=\sum_{\theta\in\Theta}\frac{p_{k,i\vert\theta}p_{\ell,i\vert\theta}\mathbb{P}(\boldsymbol{\theta}^{(k)}=\theta)\mathbb{P}(\boldsymbol{\theta}^{(\ell)}=\theta)}{\sum_{\bar{\theta}^{(k)}}p_{k,i\vert\bar{\theta}^{(k)}}q_{\bar{\theta}^{(k)}}\sum_{\bar{\theta}^{(\ell)}}p_{\ell,i\vert\bar{\theta}^{(\ell)}}q_{\bar{\theta}^{(\ell)}}}\nonumber\\
    &=\sum_{\theta\in\Theta}\boldsymbol{\pi}_{k,i}(\theta)\boldsymbol{\pi}_{\ell,i}(\theta)
\end{align}
\end{small}%
In step $(a)$ Bayes rule was utilised along with the fact that $\boldsymbol{\zeta}_{k,1:i}$ and $\boldsymbol{\zeta}_{\ell,1:i}$ are conditionally independent given $\boldsymbol{\theta}^{(k)}$ and $\boldsymbol{\theta}^{(\ell)}$, respectively. Step $(b)$ is true due to the assumption that the hypotheses are independent across agents, i.e., $
    \mathbb{P}(\boldsymbol{\theta}^{(k)},\boldsymbol{\theta}^{(\ell)})=
    \mathbb{P}(\boldsymbol{\theta}^{(k)})\mathbb{P}(\boldsymbol{\theta}^{(\ell)})$.
\qedsymb
\section{Proof of Proposition \ref{rate_localbeliefs}}
\label{a2}
For convenience for a given hypothesis $\theta\in\Theta$ we define the {\em state-specific weights} as
\begin{align}
\label{combination_weights}
    &\boldsymbol{a}_{\ell k,i}(\theta)\triangleq\frac{\mathbb{P}(\mathcal{S}_{k\ell}^{\theta}\vert\boldsymbol{\zeta}_{k,1:i},\boldsymbol{\zeta}_{\ell,1:i})}{\boldsymbol{\sigma}_{k,i}}
    =\frac{\boldsymbol{\pi}_{k,i}(\theta)\boldsymbol{\pi}_{\ell,i}(\theta)}{\boldsymbol{\sigma}_{k,i}},
\end{align}
where $\ell\in\mathcal{N}^{\star}_k$. Then, the following holds:
\begin{align}
    \boldsymbol{a}_{\ell k,i}=\sum_{\theta\in\Theta}\boldsymbol{a}_{\ell k,i}(\theta)
\end{align}
We prove the result by following the techniques used in Theorem $2$ in \cite{nedic2017fast}. We define the log ratio of local beliefs:
\begin{align}
    &\bar{\boldsymbol{\lambda}}_{k,i}(\theta)\triangleq\log\frac{\boldsymbol{\pi}_{k,i}(\theta)}{\boldsymbol{\pi}_{k,i}(\theta^{(k)})}, \quad\theta\notin\Theta^{\star}_k,k\in\mathcal{N}.
\end{align}
Using \eqref{HMM_1} we have
\begin{align}
\label{llrb_tel}
    \bar{\boldsymbol{\lambda}}_{k,i}(\theta)=\sum^i_{t=1}\boldsymbol{\mathcal{L}}_{k,t}(\theta)+\lambda_{k,0}(\theta),\quad\theta\notin\Theta^{\star}_k
\end{align}
where the last term on the right-hand side (RHS) in \eqref{llrb_tel} is equal to $0$ due to Assumption \ref{prior_beliefs} and
\begin{align}
\label{llr1}
&\boldsymbol{\mathcal{L}}_{k,i}(\theta)\triangleq\log\frac{L_k(\boldsymbol{\zeta}_{k,i}|\theta)}{L_k(\boldsymbol{\zeta}_{k,i}|\theta^{(k)})},\quad\theta\notin\Theta^{\star}_k,k\in\mathcal{N}.
\end{align}
Taking expectations in \eqref{llrb_tel} we have for every $\theta\notin\Theta^{\star}_k$
\begin{align}
\label{exp_bound_log_beliefs}
    &\mathbb{E}\{\bar{\boldsymbol{\lambda}}_{k,i}(\theta)\}=\mathbb{E}\left\{\sum^i_{t=1}\boldsymbol{\mathcal{L}}_{k,t}(\theta)\right\}\nonumber\\
    &=-id_k(\theta)\leq-i\min_{\theta\neq\theta^{(k)}}d_k(\theta).
\end{align}
Next, by Assumption \ref{non_emptysupport}, we have
\begin{align}
\label{KL_bound}
    \log \alpha\leq d_k(\theta)\leq\log \frac{1}{\alpha},\quad\forall \theta\notin\Theta^{\star}_k.
\end{align}
Let us consider the sequence of random variables $\boldsymbol{\zeta}_{k,1:i}=(\boldsymbol{\zeta}_{k,1},\ldots,\boldsymbol{\zeta}_{k,i})$. We want to establish that $\bar{\boldsymbol{\lambda}}_{k,i}(\theta)$, which is a function of $\boldsymbol{\zeta}_{k,1:i}$, has bounded differences. For all $t$ such that $1\leq t\leq i$ we have
\begin{align}
\label{McD_bound1}
    &\max_{\zeta_{k,t}}\bar{\boldsymbol{\lambda}}_{k,i}(\theta)-\min_{\zeta_{k,t}}\bar{\boldsymbol{\lambda}}_{k,i}(\theta)=\max_{\zeta_{k,t}}\log\frac{L_k(\zeta_{k,t}|\theta)}{L_k(\zeta_{k,t}|\theta^{(k)})}\nonumber\\
    &-\min_{\zeta_{k,t}}\log\frac{L_k(\zeta_{k,t}|\theta)}{L_k(\zeta_{k,t}|\theta^{(k)})}\leq\log\frac{1}{\alpha}-\log\alpha=2\log\frac{1}{\alpha}
\end{align}
where we utilized \eqref{KL_bound}. Thus, $\bar{\boldsymbol{\lambda}}_{k,i}(\theta)$ has bounded differences and as a result, we can apply McDiarmid's inequality \cite{doob1940regularity}, which states the following. Consider a sequence of random variables $\boldsymbol{\zeta}_{k,1:i}=(\boldsymbol{\zeta}_{k,1},\ldots,\boldsymbol{\zeta}_{k,i})$ and a function $g:\mathcal{Z}^i_k\to\mathbb{R}$ of bounded differences for all $1\leq t\leq i$, meaning
\begin{align}
    \sup_{\zeta_{k,t}\in\mathcal{Z}_k}g(\ldots,\zeta_{k,t},\ldots)-\inf_{\zeta_{k,t}\in\mathcal{Z}_k}g(\ldots,\zeta_{k,t},\ldots)\leq \rho_t
\end{align}
for some $\rho_t<\infty$. Then, for any $\epsilon>0$ and all $i\geq1$
\begin{align}
\label{McD}
    \mathbb{P}\Big(g(\boldsymbol{\zeta}_{k,1:i})-\mathbb{E}\{g(\boldsymbol{\zeta}_{k,1:i})\}\geq\epsilon\Big)\leq\exp\Big(-\frac{2\epsilon^2}{\sum^i_{t=1}\rho^2_t}\Big).
\end{align}
Thus, from \eqref{McD} for $\rho_t=2\log\frac{1}{\alpha}$ (the bound from \eqref{McD_bound1}) 
we have
\begin{align}
\label{bound_llrb}
    \mathbb{P}\Big(\bar{\boldsymbol{\lambda}}_{k,i}(\theta)-\mathbb{E}\{\bar{\boldsymbol{\lambda}}_{k,i}(\theta)\}\geq\epsilon\Big)\leq\exp\Big(-\frac{2\epsilon^2}{4i(\log\frac{1}{\alpha})^2}\Big).
\end{align}
Then, since $\boldsymbol{\pi}_{k,i}(\theta)\in(0,1)$ we have for all $\theta\notin\Theta^{\star}_k$
\begin{align}
\label{in1}
    \boldsymbol{\pi}_{k,i}(\theta)\leq\frac{\boldsymbol{\pi}_{k,i}(\theta)}{\boldsymbol{\pi}_{k,i}(\theta^{(k)})}=\exp\Big(\bar{\boldsymbol{\lambda}}_{k,i}(\theta)\Big).
\end{align}
Thus, for an arbitrary $\varepsilon$ we have
\begin{small}%
\begin{align}
\label{bound_belief}
    &\mathbb{P}\Big(\boldsymbol{\pi}_{k,i}(\theta)\geq\exp(\varepsilon)\Big)\leq\mathbb{P}\Big(\exp(\bar{\boldsymbol{\lambda}}_{k,i}(\theta))\geq\exp(\varepsilon)\Big)\nonumber\\
    &=\mathbb{P}\Big(\bar{\boldsymbol{\lambda}}_{k,i}(\theta)\geq\varepsilon\Big)\nonumber\\
    &\overset{\eqref{exp_bound_log_beliefs}}\leq\mathbb{P}\Big(\bar{\boldsymbol{\lambda}}_{k,i}(\theta)-\mathbb{E}\{\boldsymbol{\lambda}_{k,i}\}\geq\varepsilon+i\min_{\theta\notin\Theta^{\star}_k}d_k(\theta)\Big)
\end{align}
\end{small}%
By utilizing \eqref{bound_llrb}, \eqref{bound_belief} and by setting
\begin{align}
\varepsilon=-\frac{i}{2}\min_{\theta\notin\Theta^{\star}_k}d_k(\theta)
\end{align}
we obtain the result.
\qedsymb
\section{Proof of Theorem \ref{limiting_weights}}\label{a4}
From Proposition \ref{rate_localbeliefs} we have for every agent $k\in\mathcal{N}$:
\begin{align}
    \boldsymbol{\pi}_{k,i}(\theta)\overset{\text{a.s.}}\longrightarrow0,\quad\forall\theta\notin\Theta^{\star}_k.
\end{align}
By utilizing \eqref{combination_weights}, the above implies that
\begin{align}
    \boldsymbol{a}_{\ell k,i}(\theta)\overset{\text{a.s.}}\longrightarrow0,\quad\forall\theta\notin\Theta^{\star}_k,\forall\ell\in\mathcal{N}^{\star}_k.
\end{align}
Moreover, due to Assumption \ref{prior_beliefs} and the fact that all $\theta\in\Theta$ are observationally equivalent we have
\begin{align}
\label{equal_beliefs}
    \boldsymbol{\pi}_{k,i}(\theta)=\boldsymbol{\pi}_{k,i}(\theta'),\quad\forall\theta,\theta'\in\Theta^{\star}_k,\theta\neq\theta',\forall i\geq1.
\end{align}
Utilizing the above, since $\boldsymbol{\pi}_{k,i}$ is a probability vector we have
\begin{align}
    \boldsymbol{\pi}_{k,i}(\theta)\overset{\text{a.s.}}\longrightarrow\frac{1}{|\Theta^{\star}_k|},\quad\forall\theta\in\Theta^{\star}_k.
\end{align}
Similarly for a neighbor $\ell\in\mathcal{N}^{\star}_k$ we have that
\begin{align}
    &\boldsymbol{\pi}_{\ell,i}(\theta)\overset{\text{a.s.}}\longrightarrow0,\quad&\forall\theta\notin\Theta^{\star}_{\ell}\\
    &\boldsymbol{a}_{n\ell,i}(\theta)\overset{\text{a.s.}}\longrightarrow0,\quad&\forall\theta\notin\Theta^{\star}_{\ell},\quad\forall n\in\mathcal{N}^{\star}_{\ell}\\
    &\boldsymbol{\pi}_{\ell,i}(\theta)\overset{\text{a.s.}}\longrightarrow\frac{1}{|\Theta^{\star}_{\ell}|},\quad&\forall\theta\in\Theta^{\star}_{\ell}.
\end{align}
Utilizing the above, for an agent $k$: 
\begin{align}
&\mathbb{P}(\boldsymbol{\theta}^{(k)}=\theta,\boldsymbol{\theta}^{(\ell)}=\theta\vert\boldsymbol{\zeta}_{k,1:i},\boldsymbol{\zeta}_{\ell,1:i})\nonumber\\
&\overset{\text{a.s.}}\longrightarrow
\begin{cases}\frac{1}{|\Theta^{\star}_k|}\frac{1}{|\Theta^{\star}_{\ell}|},&\theta\in\Theta^{\star}_k\cap\Theta^{\star}_{\ell},\ell\in\mathcal{N}^{\star}_k
\\        0,&\text{otherwise}.
        \end{cases}
\end{align}
Then, \eqref{normalizing_factor} yields
\begin{align}
    &\boldsymbol{\sigma}_{k,i}\overset{\text{a.s.}}\longrightarrow
1+\sum_{\ell\in\mathcal{N}^{\star}_k}\sum_{\theta\in\Theta^{\star}_k\cap\Theta^{\star}_{\ell}}\frac{1}{|\Theta^{\star}_k|}\frac{1}{|\Theta^{\star}_{\ell}|}\nonumber\\
&=
1+
\sum_{\ell\in\mathcal{N}^{\star}_k}\frac{|\Theta^{\star}_k\cap\Theta^{\star}_{\ell}|}{|\Theta^{\star}_k||\Theta^{\star}_{\ell}|}
\end{align}
Then, from \eqref{combination_weights} we obtain the result.\qedsymb
\section{Proof of Lemma \ref{A_inf}}
\label{a8}
A strongly-connected component is a set of connected agents where information can flow from every agent to every other agent in that set and at least one agent has a self-loop (i.e., there is at least one $k\in\mathcal{N}$ such that $a_{kk}>0$) \cite{Sayed_2014}.

Since $\mathcal{G}$ is undirected we have that if $\ell\in\mathcal{N}^{\star}_k\Longleftrightarrow k\in\mathcal{N}^{\star}_{\ell}$. Moreover, due to \eqref{limiting_weightseq}, we observe that
\begin{align}
    &a_{\ell k,\infty}>0\Longleftrightarrow a_{k\ell,\infty}>0,\quad\forall k,\ell\in\mathcal{N},k\neq\ell\\
    &a_{\ell k,\infty}=0\Longleftrightarrow a_{k\ell,\infty}=0,\quad\forall k,\ell\in\mathcal{N},k\neq\ell
\end{align}
since $a_{k\ell,\infty},a_{\ell k,\infty}\propto\frac{|\Theta^{\star}_k\cap\Theta^{\star}_{\ell}|}{|\Theta^{\star}_k||\Theta^{\star}_{\ell}|}$. 
This implies that information can flow from $k$ to $\ell$ if and only if information can flow from $\ell$ to $k$ as well. This means that there is not one-directional flow of information between any two agents in the network. This implies that there can not be any path in the network that allows information to flow from any agent $k$ to any other agent $\ell$ without information from $\ell$ to flow to $k$ as well. As a result, all network is decomposed into disjoint strongly-connected components. Then, 
let $\bar{\mathcal{N}}_1,\ldots,\bar{\mathcal{N}}_S$ denote the distinct strongly connected components. Then, \eqref{Per_eq} follows from Perron-Frobenius Theorem. 
\section{Proof of Theorem \ref{public_b_c_t}}\label{a7}
First we prove a useful Lemma that provides a bound on the difference between the log-belief ratios formed by algorithm \eqref{adapt_mod}-\eqref{combine_mod} and \eqref{adapt_centralized}-\eqref{combine_centralized}. Let us first define for every possible pair $\theta,\theta'\in\Theta$ such that $\theta\neq\theta'$  the aforementioned log-belief ratios as follows:
\begin{align}
\label{lbr_mod}
&\boldsymbol{\lambda}_{k,i}(\theta,\theta')\triangleq\log\frac{\boldsymbol{\mu}_{k,i}(\theta)}{\boldsymbol{\mu}_{k,i}(\theta')},\quad k\in\mathcal{N}\\
\label{lbr_centralized}
&\boldsymbol{\lambda}^c_{k,i}(\theta,\theta')\triangleq\log\frac{\boldsymbol{\mu}^c_{k,i}(\theta)}{\boldsymbol{\mu}^c_{k,i}(\theta')},\quad k\in\mathcal{N}.
\end{align}
and the log-likelihood ratio:
\begin{align}
&\boldsymbol{\mathcal{L}}_{k,i}(\theta,\theta')\triangleq\log\frac{L_k(\boldsymbol{\zeta}_{k,i}|\theta)}{L_k(\boldsymbol{\zeta}_{k,i}|\theta')},\quad k\in\mathcal{N}
\end{align}
Then, by utilizing \eqref{adapt_mod}-\eqref{combine_mod} and \eqref{adapt_centralized}-\eqref{combine_centralized}, respectively, \eqref{lbr_mod} and \eqref{lbr_centralized} yield:
\begin{align}
&\boldsymbol{\lambda}_{k,i}(\theta,\theta')
    =\sum_{\ell\in\mathcal{N}_k}\boldsymbol{a}_{\ell k,i}\boldsymbol{\mathcal{L}}_{\ell,i}(\theta,\theta')\nonumber\\
    &+
    \sum_{\ell\in\mathcal{N}_k}\boldsymbol{a}_{\ell k,i}\boldsymbol{\lambda}_{\ell,i-1}(\theta,\theta')\\
    &\boldsymbol{\lambda}^c_{k,i}(\theta,\theta')
    =\sum_{\ell\in\mathcal{N}_k}a_{\ell k,\infty}\boldsymbol{\mathcal{L}}_{\ell,i}(\theta,\theta')\nonumber\\
    &+
    \sum_{\ell\in\mathcal{N}_k}a_{\ell k,\infty}\boldsymbol{\lambda}^c_{\ell,i-1}(\theta,\theta').
\end{align}
The above are written in matrix-vector notation as
\begin{align}
\label{log_12mvb_mod}
        &\boldsymbol{\lambda}_i(\theta,\theta')=\boldsymbol{A}^{\mathsf{T}}_i\boldsymbol{\mathcal{L}}_i(\theta,\theta')+\boldsymbol{A}^{\mathsf{T}}_i\boldsymbol{\lambda}_{i-1}(\theta,\theta')\\
    \label{log_12mvb}
        &\boldsymbol{\lambda}^c_i(\theta,\theta')=A^{\mathsf{T}}_{\infty}\boldsymbol{\mathcal{L}}_i(\theta,\theta')+A^{\mathsf{T}}_{\infty}\boldsymbol{\lambda}^c_{i-1}(\theta,\theta')
\end{align}
where
{\begin{small}{
\begin{align}
&\boldsymbol{\mathcal{L}}_i(\theta,\theta')=\Bigg[\log\frac{L_{1,i}(\boldsymbol{\zeta}_{1,i}|\theta)}{L_{1,i}(\boldsymbol{\zeta}_{1,i}|\theta')},\ldots
,\log\frac{L_{|\mathcal{N}|,i}(\boldsymbol{\zeta}_{|\mathcal{N}|,i}|\theta)}{L_{|\mathcal{N}|,i}(\boldsymbol{\zeta}_{|\mathcal{N}|,i}|\theta')}\Bigg]^{\mathsf{T}}
\end{align}
}\end{small}}%
and $\boldsymbol{\lambda}_i(\theta,\theta')=[\boldsymbol{\lambda}_{1,i}(\theta,\theta'),\ldots,\boldsymbol{\lambda}_{|\mathcal{N}|,i}(\theta,\theta')]^{\mathsf{T}}$ and $\boldsymbol{\lambda}^c_i(\theta,\theta')=[\boldsymbol{\lambda}^c_{1,i}(\theta,\theta'),\ldots,\boldsymbol{\lambda}^c_{|\mathcal{N}|,i}(\theta,\theta')]^{\mathsf{T}}$.
\begin{Lem} \label{bounded_log_belief_ratios}{\bf (Bounded difference of log-belief ratios)}. 
The difference between the log-belief ratios is bounded for all $i\geq1$ and for all $\theta\neq\theta'$, $\theta,\theta'\in\Theta$, meaning
\begin{align}
    \lim_{i\to\infty}\sup\mathbb{E}\{||\boldsymbol{\lambda}_i(\theta,\theta')-\boldsymbol{\lambda}^c_i(\theta,\theta')||_{\infty}\}\leq y<\infty.
\end{align}
\end{Lem}
\begin{proof}
By utilizing \eqref{log_12mvb_mod} and \eqref{log_12mvb}, the difference in the log-belief ratios between the two algorithms for any two distinct $\theta,\theta'\in\Theta$ is given by
\begin{small}
\begin{align}
    &\boldsymbol{\lambda}_i(\theta,\theta')-\boldsymbol{\lambda}^c_i(\theta,\theta')\nonumber\\
    &=(\boldsymbol{A}^{\mathsf{T}}_i-A^{\mathsf{T}}_{\infty})\boldsymbol{\mathcal{L}}_i(\theta,\theta')+\boldsymbol{A}^{\mathsf{T}}_i\boldsymbol{\lambda}_{i-1}(\theta,\theta')\nonumber\\
    &-A^{\mathsf{T}}_{\infty}\boldsymbol{\lambda}^c_{i-1}(\theta,\theta')+\boldsymbol{A}^{\mathsf{T}}_i\boldsymbol{\lambda}^c_{i-1}(\theta,\theta')-\boldsymbol{A}^{\mathsf{T}}_i\boldsymbol{\lambda}^c_{i-1}(\theta,\theta')\nonumber\\
    &=(\boldsymbol{A}^{\mathsf{T}}_i-A^{\mathsf{T}}_{\infty})\boldsymbol{\mathcal{L}}_i(\theta,\theta')+\boldsymbol{A}^{\mathsf{T}}_i(\boldsymbol{\lambda}_{i-1}(\theta,\theta')-\boldsymbol{\lambda}^c_{i-1}(\theta,\theta'))\nonumber\\
    &+(\boldsymbol{A}^{\mathsf{T}}_i-A^{\mathsf{T}}_{\infty})\boldsymbol{\lambda}^c_{i-1}(\theta,\theta')
\end{align}
\end{small}%
Taking the  $L_{\infty}$-norm we have
\begin{small}
\begin{align}
    &||\boldsymbol{\lambda}_i(\theta,\theta')-\boldsymbol{\lambda}^c_i(\theta,\theta')||_{\infty}\overset{(a)}\leq
    ||(\boldsymbol{A}^{\mathsf{T}}_i-A^{\mathsf{T}}_{\infty})||_{\infty}||\boldsymbol{\mathcal{L}}_i(\theta,\theta')||_{\infty}\nonumber\\
    &+||\boldsymbol{A}^{\mathsf{T}}_i||_{\infty}||(\boldsymbol{\lambda}_{i-1}(\theta,\theta')-\boldsymbol{\lambda}^c_{i-1}(\theta,\theta'))||_{\infty}\nonumber\\
    &+||(\boldsymbol{A}^{\mathsf{T}}_i-A^{\mathsf{T}}_{\infty})||_{\infty}||\boldsymbol{\lambda}^c_{i-1}(\theta,\theta')||_{\infty}\nonumber\\
    &\overset{(b)}=    ||(\boldsymbol{A}^{\mathsf{T}}_i-A^{\mathsf{T}}_{\infty})||_{\infty}||\boldsymbol{\mathcal{L}}_i(\theta,\theta')||_{\infty}\nonumber\\
    &+||\boldsymbol{\lambda}_{i-1}(\theta,\theta')-\boldsymbol{\lambda}^c_{i-1}(\theta,\theta')||_{\infty}\nonumber\\
    &+||(\boldsymbol{A}^{\mathsf{T}}_i-A^{\mathsf{T}}_{\infty})||_{\infty}||\boldsymbol{\lambda}^c_{i-1}(\theta,\theta')||_{\infty}
\end{align}
\end{small}%
where $||x||_{\infty}$ corresponds to the {\em maximum (absolute) row sum norm} (induced by $L_{\infty}$ vector norm) if $x$ is a matrix. $(a)$ is true by the sub-multiplicative property of vector induced norms ({\em see} \cite{horn2012matrix} Theorem 5.6.2 property b) and $(b)$ due to the fact that $A^{\mathsf{T}}$ is right-stochastic and as a result $||A^{\mathsf{T}}_i||_{\infty}=1$ for all $i$. Expanding the above we get
\begin{small}
\begin{align}
\label{rec_in2}
    &||\boldsymbol{\lambda}_i(\theta,\theta')-\boldsymbol{\lambda}^c_i(\theta,\theta')||_{\infty}\nonumber\\
    &
    \leq\sum_{t=1}^i||(\boldsymbol{A}^{\mathsf{T}}_t-A^{\mathsf{T}}_{\infty})||_{\infty}||\boldsymbol{\mathcal{L}}_t(\theta,\theta')||_{\infty}\nonumber\\
    &+||\lambda_0(\theta,\theta')-\lambda^c_0(\theta,\theta')||_{\infty}\nonumber\\
    &+\sum_{t=1}^i||(\boldsymbol{A}^{\mathsf{T}}_t-A^{\mathsf{T}}_{\infty})||_{\infty}||\boldsymbol{\lambda}^c_{t-1}(\theta,\theta')||_{\infty}\nonumber\\
    &\overset{(a)}=\sum_{t=1}^i||(\boldsymbol{A}^{\mathsf{T}}_t-A^{\mathsf{T}}_{\infty})||_{\infty}||\boldsymbol{\mathcal{L}}_t(\theta,\theta')||_{\infty}\nonumber\\
    &+\sum_{t=1}^i||(\boldsymbol{A}^{\mathsf{T}}_t-A^{\mathsf{T}}_{\infty})||_{\infty}||\boldsymbol{\lambda}^c_{t-1}(\theta,\theta')||_{\infty}\nonumber\\
    &\overset{(b)}\leq L\sum_{t=1}^i||(\boldsymbol{A}^{\mathsf{T}}_t-A^{\mathsf{T}}_{\infty})||_{\infty}\nonumber\\
    &+\sum_{t=1}^i||(\boldsymbol{A}^{\mathsf{T}}_t-A^{\mathsf{T}}_{\infty})||_{\infty}||\boldsymbol{\lambda}^c_{t-1}(\theta,\theta')||_{\infty}
\end{align}
\end{small}%
where $(a)$ is true due to Assumption \ref{prior_beliefs} and $(b)$ because of the fact that $|\boldsymbol{\mathcal{L}}_t(\theta,\theta')|$ is bounded by $L=|\log\alpha|$. Regarding the second term in \eqref{rec_in2}, by iterating \eqref{log_12mvb} and using the sub-multiplicative property, we have
\begin{align}
\label{rec_in3}
   &\sum_{t=1}^i||(\boldsymbol{A}^{\mathsf{T}}_t-A^{\mathsf{T}}_{\infty})||_{\infty}||\boldsymbol{\lambda}^c_{t-1}(\theta,\theta')||_{\infty}\nonumber\\
    &\leq\sum_{t=1}^i||(\boldsymbol{A}^{\mathsf{T}}_t-A^{\mathsf{T}}_{\infty})||_{\infty}\sum^{t-1}_{t'=1}||(A^{\mathsf{T}}_{\infty})^{t'}||_{\infty}||\boldsymbol{\mathcal{L}}_{t-t'}(\theta,\theta')||_{\infty}\nonumber\\
    &+||(A^{\mathsf{T}}_{\infty})^{t-1}||_{\infty}||\lambda_0(\theta,\theta')||_{\infty}\nonumber\\
    &\overset{(a)}=\sum_{t=1}^i||(\boldsymbol{A}^{\mathsf{T}}_t-A^{\mathsf{T}}_{\infty})||_{\infty}\sum^{t-1}_{t'=1}||\boldsymbol{\mathcal{L}}_{t'}(\theta,\theta')||_{\infty}\nonumber\\
    &\overset{(b)}\leq \sum_{t=1}^i||(\boldsymbol{A}^{\mathsf{T}}_t-A^{\mathsf{T}}_{\infty})||_{\infty}\sum^{t-1}_{t'=1}L\nonumber\\
    &=L\sum_{t=1}^i||(\boldsymbol{A}^{\mathsf{T}}_t-A^{\mathsf{T}}_{\infty})||_{\infty}(t-1)
\end{align}
where $(a)$ is true due to Assumption \ref{prior_beliefs} and the fact that $||(A^{\mathsf{T}}_{\infty})^{t'}||_{\infty}=1$ for all $t'$. $(b)$ holds for the same reason as in \eqref{rec_in2}. Combining  \eqref{rec_in2} and \eqref{rec_in3} we get
\begin{align}
\label{rec_in4}
    &||\boldsymbol{\lambda}_i(\theta,\theta')-\boldsymbol{\lambda}^c_i(\theta,\theta')||_{\infty}
    \leq
    L\sum_{t=1}^i||(\boldsymbol{A}^{\mathsf{T}}_t-A^{\mathsf{T}}_{\infty})||_{\infty}\nonumber\\
    &+\sum_{t=1}^i||(\boldsymbol{A}^{\mathsf{T}}_t-A^{\mathsf{T}}_{\infty})||_{\infty}\sum^i_{t'=t}||\boldsymbol{\mathcal{L}}_{t'}||_{\infty}
    \nonumber\\
    &\leq L\sum_{t=1}^i||(\boldsymbol{A}^{\mathsf{T}}_t-A^{\mathsf{T}}_{\infty})||_{\infty}\nonumber\\
    &+L\sum_{t=1}^i||(\boldsymbol{A}^{\mathsf{T}}_t-A^{\mathsf{T}}_{\infty})||_{\infty}(t-1)
\end{align}
Let us bound  $\mathbb{E}\{\sum^i_{t=1}|\boldsymbol{a}_{\ell k,t}-a_{\ell k,\infty}|\}$. We have for every $k\in\mathcal{N},\ell\in\mathcal{N}^{\star}_k$ and for all $\theta\notin\Theta^{\star}_k\cap\Theta^{\star}_{\ell}$
\begin{small}
\begin{align}
    &\mathbb{E}\left\{\sum^i_{t=1}|\boldsymbol{a}_{\ell k,t}(\theta)-0|\right\}=\mathbb{E}\Bigg\{\sum^i_{t=1}|\boldsymbol{a}_{\ell k,t}(\theta)\mathbb{I}_{\{\boldsymbol{a}_{\ell k,t}\geq\exp(c_{\ell k}t)\}}|\nonumber\\
    &+\sum^i_{t=1}|\boldsymbol{a}_{\ell k,t}(\theta)\mathbb{I}_{\{\boldsymbol{a}_{\ell k,t}(\theta)<\exp(c_{\ell k}t)\}}|\Bigg\}\nonumber\\
    &\overset{(a)}\leq
    \sum^i_{t=1}1\times\exp(-d_{\ell k}t)+\sum^i_{t=1}\exp(-c_{\ell k}t)\times1
    \end{align}
    \end{small}%
where in $(a)$ we utilized part $1)$ of Lemma \ref{rate_weights} to upper bound the value of $|\boldsymbol{a}_{\ell k,t}(\theta)-0|=\boldsymbol{a}_{\ell k,t}(\theta)$. Moreover, following the same rationale as above, from part $2)$ of Lemma \ref{rate_weights} we have for every $k\in\mathcal{N},\ell\in\mathcal{N}^{\star}_k$ and for all $\theta\in\Theta^{\star}_k\cap\Theta^{\star}_{\ell}$
    \begin{align}
    &\mathbb{E}\left\{\sum^i_{t=1}|\boldsymbol{a}_{\ell k,t}(\theta)-a_{\ell k,\infty}(\theta)|\right\}\nonumber\\
    &=\mathbb{E}\Bigg\{\sum^i_{t=1}|\boldsymbol{a}_{\ell k,t}(\theta)-a_{\ell k,\infty}|\mathbb{I}_{\{|\boldsymbol{a}_{\ell k,t}(\theta)|\geq\alpha_{\ell k}'\exp(-c_{\ell k}'t)\}}\nonumber\\
    &+\sum^i_{t=1}|\boldsymbol{a}_{\ell k,t}(\theta)-a_{\ell k,\infty}|\mathbb{I}_{\{|\boldsymbol{a}_{\ell k,t}(\theta)|<\alpha_{\ell k}'\exp(-c_{\ell k}'t)\}}\Bigg\}\nonumber\\
    &\leq\sum^i_{t=1}1\times b_{\ell k}'\exp(-d_{\ell k}'t)+\sum^i_{t=1}\alpha_{\ell k}'\exp(-c_{\ell k}'t)\times1.
\end{align}
Moreover, we have
\begin{small}
\begin{align}
    &\sum^i_{t=1}||(\boldsymbol{A}^{\mathsf{T}}_t-A^{\mathsf{T}}_{\infty})||_{\infty}\leq\sum^i_{t=1}\max_{k\in\mathcal{N}}\Bigg\{\sum_{\ell\neq k}|\boldsymbol{a}_{\ell k,t}-a_{\ell k,\infty}|\nonumber\\
    &+|1-\sum_{\ell\neq k}\boldsymbol{a}_{\ell k,t}-1+\sum_{\ell\neq k}a_{\ell k,\infty}|\Bigg\}\nonumber\\
    &\leq2\sum^i_{t=1}\max_{k\in\mathcal{N}}\left\{\sum_{\ell\neq k}|\boldsymbol{a}_{\ell k,t}-a_{\ell k,\infty}|\right\}
    \nonumber\\
    &\leq2\sum^i_{t=1}\max_{k\in\mathcal{N}}\Bigg\{\sum_{\ell\neq k}(\sum_{\theta\notin\Theta^{\star}_k\cap\Theta^{\star}_{\ell}}|\boldsymbol{a}_{\ell k,t}(\theta)-0|\nonumber\\
    &+\sum_{\theta\in\Theta^{\star}_k\cap\Theta^{\star}_{\ell}}|\boldsymbol{a}_{\ell k,t}(\theta)-a_{\ell k,\infty}(\theta)|)\Bigg\}\nonumber\\
    &\overset{(a)}\leq2\sum^i_{t=1}\sum_{k\in\mathcal{N}}\Bigg(\sum_{\ell\neq k}\Big(\sum_{\theta\notin\Theta^{\star}_k\cap\Theta^{\star}_{\ell}}|\boldsymbol{a}_{\ell k,t}(\theta)-0|\nonumber\\
    &+\sum_{\theta\in\Theta^{\star}_k\cap\Theta^{\star}_{\ell}}|\boldsymbol{a}_{\ell k,t}(\theta)-a_{\ell k,\infty}(\theta)|\Big)\Bigg).
\end{align}
\end{small}%
where $(a)$ is true due to the fact that all the summation terms are positive and thus the total sum of the elements is greater than the maximum element. By taking expectation, we have
\begin{small}
\begin{align}
\label{sfa}
    &\mathbb{E}\left\{\sum^i_{t=1}||(\boldsymbol{A}^{\mathsf{T}}_t-A^{\mathsf{T}}_{\infty})||_{\infty}\right\}\nonumber\\
    &\leq
    2\sum^i_{t=1}\sum_{k\in\mathcal{N}}\Bigg(\sum_{\ell\neq k}\Big(\sum_{\theta\notin\Theta^{\star}_k\cap\Theta^{\star}_{\ell}}\exp(-d_{\ell k}t)+\exp(c_{\ell k}t)\nonumber\\
    &+\sum_{\theta\in\Theta^{\star}_k\cap\Theta^{\star}_{\ell}}b_{\ell k,t}'\exp(-d_{\ell k}'t)+\alpha_{\ell k,t}'\exp(-c_{\ell k}'t)\Big)\Bigg)=\xi_i
\end{align}
\end{small}%
We have
\begin{align}
    \xi_{\infty}\triangleq\lim_{i\to\infty}\xi_i<\infty.
\end{align}
because all series appearing in \eqref{sfa} are convergent geometric series. Then, \eqref{rec_in4} yields
\begin{small}
\begin{align}
\label{rec_in5}
    &\mathbb{E}\{||\boldsymbol{\lambda}_i(\theta,\theta')-\boldsymbol{\lambda}^c_i(\theta,\theta')||_{\infty}\}\nonumber\\
    &\leq\hspace{-1mm}\mathbb{E}\hspace{-1mm}\left\{\hspace{-1mm}L\sum_{t=1}^i||(\boldsymbol{A}^{\mathsf{T}}_t-A^{\mathsf{T}}_{\infty})||_{\infty}\hspace{-1mm}+\hspace{-1mm}L\sum_{t=1}^i||(\boldsymbol{A}^{\mathsf{T}}_t-A^{\mathsf{T}}_{\infty})||_{\infty}(t-1)\hspace{-1mm}\right\}\nonumber\\
    &\leq L\xi_i+2L\Bigg(\sum^i_{t=1}\sum_{k\in\mathcal{N}}\Big(\sum_{\ell\neq k}\big(\sum_{\theta\notin\Theta^{\star}_k\cap\Theta^{\star}_{\ell}}\exp(-d_{\ell k}t)(t-1)\nonumber\\
    &+\exp(-c_{\ell k}t)(t-1)+\sum_{\theta\in\Theta^{\star}_k\cap\Theta^{\star}_{\ell}}\alpha_{\ell k}'\exp(-c_{\ell k}'t)(t-1)\nonumber\\
    &+b_{\ell k}'\exp(-d_{\ell k}'t)(t-1)\big)\Big)\Bigg).
\end{align}
\end{small}%
Let us study the series appearing in the second term on the RHS of the above inequality. Performing the ratio test for each one of them we have
\begin{align}
    &r=\lim_{t\to\infty}\frac{\exp(-v(t+1))t}{\exp(-vt)(t-1)}\nonumber\\
    &=\exp(-v)<1,\quad\forall v>0.
\end{align}
Thus, all series on LHS of \eqref{rec_in5} converge, implying
\begin{align}
    \lim_{i\to\infty}\sup\mathbb{E}\{||\boldsymbol{\lambda}_i(\theta,\theta')-\boldsymbol{\lambda}^c_i(\theta,\theta')||_{\infty}\}\leq y<\infty.
\end{align}
\end{proof}
Now we prove Theorem \ref{public_b_c_t}. We characterize the asymptotic behavior of $\boldsymbol{\mu}^c_{k,i}$ and then we use Lemma \ref{bounded_log_belief_ratios} to characterize the behavior of $\boldsymbol{\mu}_{k,i}$. 
Expanding \eqref{log_12mvb} yields
\begin{align}
\label{eqq}
    &\boldsymbol{\lambda}^c_i(\theta,\theta')=
    \sum^i_{t=1}(A^{\mathsf{T}}_{\infty})^{i-t+1}\boldsymbol{\mathcal{L}}_t(\theta,\theta')\nonumber\\
    &+(A^{\mathsf{T}}_{\infty})^i\lambda_0(\theta,\theta'),\quad \theta\neq\theta'
\end{align}
Since the prior beliefs are uniform (Assumption \ref{prior_beliefs}) the second term in \eqref{eqq} is equal to $0$. 
Then, by adding and subtracting 
$\bar{A}^{\mathsf{T}}_{\infty}$ (defined in \eqref{Per_eq} and by dividing by $i$ and taking the limit as $i\to\infty$ \eqref{eqq} yields
\begin{align}
&\lim_{i\to\infty}\hspace{-1mm}\frac{1}{i}\boldsymbol{\lambda}^c_i(\theta,\theta')\hspace{-1mm}=\hspace{-1mm}\lim_{i\to\infty}\hspace{-1mm}\frac{1}{i}\hspace{-1mm}\sum^i_{t=1}\hspace{-1mm}\left((A^{\mathsf{T}}_{\infty})^{i-t+1}-\bar{A}^{\mathsf{T}}_{\infty}\right)\hspace{-1mm}\boldsymbol{\mathcal{L}}_t(\theta,\theta')\nonumber\\
&+\lim_{i\to\infty}\frac{1}{i}\sum^i_{t=1}\bar{A}^{\mathsf{T}}_{\infty}\boldsymbol{\mathcal{L}}_t(\theta,\theta')
\end{align}
Following the same arguments used in the proof of Lemma $8$ in \cite{bordignon2020socialarxiv} 
we can show that the first term on the RHS of the above expression goes to $0$ a.s. 
Then, from the strong law of large numbers we have that
\begin{align}
    \lim_{i\to\infty}\frac{1}{i}\boldsymbol{\lambda}^c_i(\theta,\theta')
    \overset{\text{a.s.}}\longrightarrow \bar{A}^{\mathsf{T}}_{\infty}
    \mathbb{E}\{\boldsymbol{
    \mathcal{L}}_t(\theta,\theta')\}
\end{align}
From Lemma \ref{A_inf}, the above implies that for every $k\in\bar{\mathcal{N}}_s$
\begin{align}
    &\lim_{i\to\infty}\frac{1}{i}\boldsymbol{\lambda}^c_{k,i}(\theta,\theta')
    \overset{\text{a.s.}}\longrightarrow \sum_{\ell\in\bar{\mathcal{N}}_s}p_s({\ell})
    \mathbb{E}\{\boldsymbol{
    \mathcal{L}}_{\ell,t}(\theta,\theta')\}\nonumber\\
    &=\sum_{\ell\in\bar{\mathcal{N}}_s}p_s({\ell})
    (d_{\ell}(\theta')-d_{\ell}(\theta))
\end{align}
Observing the above, we conclude that for any $\theta\notin\bar{\Theta}^{\star}_s$ and for any $\theta'\in\bar{\Theta}^{\star}_s$ we have
\begin{align}
    &\lim_{i\to\infty}\frac{1}{i}\boldsymbol{\lambda}^c_i(\theta,\theta')
    \overset{\text{a.s.}}\longrightarrow-(C_s(\theta^{\star}))-C_s(\theta))
\end{align}
which is $C_s(\theta^{\star}_s)-C_s(\theta)>0$ from the definition of $\theta^{\star}_s$ in \eqref{confidence_max}. Also, note that $C_s(\theta^{\star})-C_s(\theta)$ is finite due to Assumption \ref{non_emptysupport}.

Since $\frac{\boldsymbol{\lambda}^c_{k,i}(\theta,\theta')}{i}$ converges to a finite negative value, we have that $\boldsymbol{\lambda}^c_{k,i}(\theta,\theta')$ diverges to $-\infty$, which in turn implies that $\boldsymbol{\mu}^c_{k,i}(\theta)\overset{a.s.}\longrightarrow0$ for all $\theta\notin\bar{\Theta}^{\star}_s$.

Then, we have
\begin{align}
    &\mathbb{E}\left\{\left|\left|\frac{\boldsymbol{\lambda}_i(\theta,\theta')}{i}-\frac{\boldsymbol{\lambda}^c_i(\theta,\theta')}{i}\right|\right|_{\infty}\right\}\nonumber\\
    &=\frac{1}{i}\mathbb{E}\{||\boldsymbol{\lambda}_i(\theta,\theta')-\boldsymbol{\lambda}^c_i(\theta,\theta')||_{\infty}\}
\end{align}
From Lemma \ref{bounded_log_belief_ratios} by taking the limit as $i$ goes to $\infty$ we get
\begin{align}
    &\lim_{i\to\infty}\mathbb{E}\left\{\left|\left|\frac{\boldsymbol{\lambda}_i(\theta,\theta')}{i}-\frac{\boldsymbol{\lambda}^c_i(\theta,\theta')}{i}\right|\right|_{\infty}\right\}\nonumber\\
    &=\lim_{i\to\infty}\frac{1}{i}\mathbb{E}\{||\boldsymbol{\lambda}_i(\theta,\theta')-\boldsymbol{\lambda}^c_i(\theta,\theta')||_{\infty}\}\leq\lim_{i\to\infty}\frac{y}{i}=0
\end{align}
which implies that
\begin{align}
    \frac{\boldsymbol{\lambda}_i(\theta,\theta')}{i}\overset{\text{P.}}\longrightarrow\frac{\boldsymbol{\lambda}^c_i(\theta,\theta')}{i}
\end{align}
Then, since
\begin{align}
    \frac{\boldsymbol{\lambda}^c_i(\theta,\theta')}{i}\overset{\text{a.s.}}\longrightarrow-(C_s(\theta^{\star})-C_s(\theta))
\end{align}
we have that
\begin{align}
    \frac{\boldsymbol{\lambda}_i(\theta,\theta')}{i}\overset{\text{P.}}\longrightarrow-(C_s(\theta^{\star})-C_s(\theta)).
\end{align}
This implies
\begin{align}
    \boldsymbol{\mu}_{k,i}(\theta)\overset{\text{P.}}\longrightarrow0,\quad\forall\theta\notin\bar{\Theta}^{\star}_s.
\end{align}
For the part $2)$ of the Theorem, we have that $\boldsymbol{\mu}_{k,i}$ is a probability vector and thus its entries must sum up to $1$. Then, if $\bar{\Theta}^{\star}_s=\{\theta^{(k)}\}$, from the first part of the Theorem we have that  $\boldsymbol{\mu}_{k,i}(\theta)\overset{P.}\longrightarrow0$ for all $\theta\neq\theta^{(k)}$ because $\bar{\Theta}^{\star}_s$ contains only $\theta^{(k)}$ and thus, $\boldsymbol{\mu}_{k,i}(\theta^{(k)})\overset{P.}\longrightarrow1$.
\hfill\qedsymb
\section{Auxiliary Results}
\label{ap_aux}
\begin{Lem}
\label{in_form}
Let two random variables $\boldsymbol{x}_i,\boldsymbol{y}_i$ such that:
\begin{align}
\label{x_i}
    &\mathbb{P}\left(|\boldsymbol{x}_i-x_{\infty}|\geq c_x\exp(-a_xi)\right)\leq d_x\exp(-b_xi)\\
    \label{y_i}
    &\mathbb{P}\left(|\boldsymbol{y}_i-y_{\infty}|\geq c_y\exp(-a_yi)\right)\leq d_y\exp(-b_yi)
\end{align}
for some $a_x,b_x,a_y,b_y>0$ and $x_{\infty},y_{\infty}\in\mathbb{R}$. Then:
\begin{align}
    \label{f2}
    &\hspace{-2mm}\mathbb{P}\left(|\boldsymbol{x}_i\hspace{-1mm}+\hspace{-1mm}\boldsymbol{y}_i\hspace{-1mm}-\hspace{-1mm}x_{\infty}\hspace{-1mm}-\hspace{-1mm}y_{\infty}|\hspace{-1mm}\geq\hspace{-1mm}\bar{c}\exp(-\bar{a}i)\right)
    \leq\bar{d}\exp(-\bar{b}i)\\
    \label{f3}
    &\hspace{-2mm}\mathbb{P}\left(|\boldsymbol{x}_i\boldsymbol{y}_i-x_{\infty}y_{\infty}|\geq\bar{c}\exp(-\bar{a}i)\right)\leq\bar{d}\exp(-\bar{b}i)
\end{align}
for some $\bar{a},\bar{b},\bar{c},\bar{d}>0$.

If further $\boldsymbol{x}_i,x_{\infty}\geq1$, then:
\begin{align}
    \label{f1}
    &\mathbb{P}\left(|\boldsymbol{x}_i^{-1}-x_{\infty}^{-1}|\geq\ \bar{c}\exp(-\bar{a}i)\right)\leq \bar{d}\exp(-\bar{b}i)
\end{align}
for some $\bar{a},\bar{b},\bar{c},\bar{d}>0$.
\end{Lem}
\begin{proof}
Let us define the following events:
\begin{align}
    &\mathcal{A}\triangleq\{|\boldsymbol{x}_i-x_{\infty}|\geq c_x\exp(-a_xi)\}\\
    &\mathcal{B}\triangleq\{|\boldsymbol{y}_i-y_{\infty}|\geq c_y\exp(-a_yi)\}\\
    &\mathcal{C}(\bar{a},\bar{c})\triangleq\{|\boldsymbol{x}^{-1}_i-x^{-1}_{\infty}|\geq \bar{c}\exp(-\bar{a}i)\}\\
    &\mathcal{D}(\bar{a},\bar{c})\triangleq\{|\boldsymbol{x}_i+\boldsymbol{y}_i-x_{\infty}-y_{\infty}|\geq \bar{c}\exp(-\bar{a}i)\}\\
    &\mathcal{E}(\bar{a},\bar{c})\triangleq\{|\boldsymbol{x}_i\boldsymbol{y}_i-x_{\infty}y_{\infty}|\geq \bar{c}\exp(-\bar{a}i)\}
\end{align}
for some $a_x,a_y,c_x,c_y,\bar{a},\bar{c}>0$.

We prove first \eqref{f1}. We have
\begin{align}
    |\boldsymbol{x}^{-1}_i-x^{-1}_{\infty}|=\frac{|\boldsymbol{x}_i-x_{\infty}|}{|\boldsymbol{x}_ix_{\infty}|}\leq|\boldsymbol{x}_i-x_{\infty}|
\end{align}
because $|\boldsymbol{x}_ix_{\infty}|\geq1$. The above implies
\begin{align}
    \bar{\mathcal{A}}\Rightarrow\bar{\mathcal{C}}(a_x,c_x)
\end{align}
where $\bar{\mathcal{A}}$ stands for the complement of an event $\mathcal{A}$. The above in turn implies
\begin{align}
    &\mathbb{P}(\bar{\mathcal{C}}(a_x,c_x))\geq\mathbb{P}(\bar{\mathcal{A}})\nonumber\\
    &\Leftrightarrow\mathbb{P}(\mathcal{C}(a_x,c_x))\leq\mathbb{P}(\mathcal{A})\overset{\eqref{x_i}}\leq d_x\exp(-b_xi).
\end{align}
The above implies that \eqref{f1} holds for $\boldsymbol{x}_i,x_{\infty}\geq1$ and for $\bar{a}=a_x,\bar{b}=b_x,\bar{c}=c_x,\bar{d}=d_x$.

We move on to prove \eqref{f2}. From the triangle inequality, we have
\begin{align}
\label{v1}
    |\boldsymbol{x}_i+\boldsymbol{y}_i-x_{\infty}-y_{\infty}|\leq|\boldsymbol{x}_i-x_{\infty}|+|\boldsymbol{y}_i-y_{\infty}|
\end{align}
The above implies that
\begin{align}
\label{v2}
    &\bar{\mathcal{A}}\cap\bar{\mathcal{B}}\nonumber\\
    &\Rightarrow |\boldsymbol{x}_i+\boldsymbol{y}_i-x_{\infty}-y_{\infty}|\nonumber\\
    &<c_x\exp(-a_xi)+c_y\exp(-a_yi)\nonumber\\
    &<(c_x+c_y)\exp(-\min\{a_x,a_y\}i).
\end{align}
The above implies
\begin{align}
\label{v3}
    \bar{\mathcal{A}}\cap\bar{\mathcal{B}}\Rightarrow\bar{\mathcal{D}}(-\min\{a_x,a_y\}i,c_x+c_y)
\end{align}
which in turn implies
\begin{align}
\label{v4}
    &\mathbb{P}(\bar{\mathcal{D}}(\min\{a_x,a_y\},c_x+c_y))\geq\mathbb{P}(\bar{\mathcal{A}}\cap\bar{\mathcal{B}})\nonumber\\
    &\Leftrightarrow\mathbb{P}(\mathcal{D}(\min\{a_x,a_y\}i,c_x+c_y))\nonumber\\
    &\overset{(a)}\leq\mathbb{P}(\mathcal{A}\cup \mathcal{B})\overset{(b)}\leq\mathbb{P}(\mathcal{A})+\mathbb{P}(\mathcal{B})\nonumber\\
    &\overset{\eqref{x_i},\eqref{y_i}}\leq d_x\exp(-a_xi)+d_y\exp(-a_yi)
\end{align}
where in $(a)$ De Morgan's law \cite{hurley2014concise} was utilized and in $(b)$ we used the union bound. The above implies that \eqref{f2} holds for $\bar{a}=\min\{a_x,a_y\},\,\bar{b}=\min\{b_x,b_y\},\,\bar{c}=c_x+c_y,\,\bar{d}=d_x+d_y$.

Finally, we prove \eqref{f3}. We have.
\begin{align}
    &|\boldsymbol{x}_i\boldsymbol{y}_i-x_{\infty}y_{\infty}|=\nonumber\\
    &|(\boldsymbol{x}_i-x_{\infty})(\boldsymbol{y}_i-y_{\infty})+x_{\infty}(\boldsymbol{y}_i-y_{\infty})+y_{\infty}(\boldsymbol{x}_i-x_{\infty})|\nonumber\\
    &\leq|\boldsymbol{x}_i-x_{\infty}||\boldsymbol{y}_i-y_{\infty}|\nonumber\\
    &+|x_{\infty}||\boldsymbol{y}_i-y_{\infty}|+|y_{\infty}||\boldsymbol{x}_i-x_{\infty}|
\end{align}
By working in the same way as in \eqref{v3}, \eqref{v4}, from the above we obtain
\begin{align}
    \bar{\mathcal{A}}\cap\bar{\mathcal{B}}\Rightarrow\bar{\mathcal{E}}(\min\{a_x,a_y\},c_xc_y+x_{\infty}c_y+y_{\infty}c_x)
\end{align}
which implies
\begin{align}
    &\mathbb{P}(\bar{\mathcal{E}}(\min\{a_x,a_y\},c_xc_y+x_{\infty}c_y+y_{\infty}c_x))\nonumber\\
    &\geq\mathbb{P}(\bar{\mathcal{A}}\cap\bar{\mathcal{B}})\nonumber\\
    &\Leftrightarrow\mathbb{P}(\mathcal{E}(\min\{a_x,a_y\},c_xc_y+x_{\infty}c_y+y_{\infty}c_x))\nonumber\\
    &\leq\mathbb{P}(\mathcal{A}\cup\mathcal{B})\leq\mathbb{P}(\mathcal{A})+\mathbb{P}(\mathcal{B})\nonumber\\
    &\leq d_x\exp(-a_xi)+d_y\exp(-a_yi)
\end{align}
The above implies that \eqref{f3} holds for $\bar{a}=\min\{a_x,a_y\},\,\bar{b}=\min\{b_x,b_y\},\,\bar{c}=c_xc_y+x_{\infty}c_y+y_{\infty}c_x),\,\bar{d}=d_x+d_y$.
\end{proof}
The following auxiliary Lemma characterizes the evolution of the adaptive weights.
\begin{Lem}
\label{rate_weights} {\bf (Rate of convergence of the combination weights
)}. Under Assumptions \ref{prior_beliefs} and \ref{non_emptysupport} the following hold:
\begin{enumerate}
    \item For every $\theta\notin\Theta^{\star}_k\cap\Theta^{\star}_{\ell}$ and $\ell\in\mathcal{N}^{\star}_k$ the following holds:
\begin{align}
    &\mathbb{P}\left(\boldsymbol{a}_{\ell k,i}(\theta)\geq\exp\left(-c_{\ell k}i\right)\right)\leq\exp\left(-d_{\ell k}i\right)
\end{align}
where
\begin{align}
&c_{\ell k}\triangleq\frac{1}{2}\min_{\theta\notin\Theta^{\star}_k\cap\Theta^{\star}_{\ell}}\sum_{m\in\{k,\ell\}}d_m(\theta)\\
&d_{\ell k}\triangleq\frac{\min_{\theta\notin\Theta^{\star}_k\cap\Theta^{\star}_{\ell}}\sum_{m\in\{k,\ell\}}d^2_m(\theta)}{32(\log\alpha)^2}
\end{align}
\item For every $\theta\in\Theta^{\star}_k\cap\Theta^{\star}_{\ell}$ and for all $k\in\mathcal{N},\ell\in\mathcal{N}^{\star}_k$ $\boldsymbol{a}_{\ell k,i}(\theta)$ the following is true: 
        \begin{align}
        \label{bound_intersectionweights}
     &\mathbb{P}\left(|\boldsymbol{a}_{\ell k,i}(\theta)-a_{\ell k,\infty}(\theta)|\geq \alpha_{\ell k}'\exp(-c_{\ell k}'i)\right)\nonumber\\
     &\leq b_{\ell k}'\exp(-d_{\ell k}'i)
        \end{align}
        for some $\alpha_{\ell k}',c_{\ell k}',b_{\ell k}',d_{\ell k}'>0$ and for all $i\geq1$.
\end{enumerate}
\end{Lem}
\begin{proof}We will prove first part $1)$ of the Lemma. Since $\boldsymbol{\sigma}_{k,i}\geq1,\quad\forall i\geq1,k\in\mathcal{N},\ell\in\mathcal{N}^{\star}_k$, for all $\theta\notin\Theta^{\star}_k\cap\Theta^{\star}_{\ell}$ we have from \eqref{combination_weights}:
\begin{align}
    \boldsymbol{a}_{\ell k,i}(\theta)\leq\boldsymbol{\pi}_{k,i}(\theta)\boldsymbol{\pi}_{\ell,i}(\theta)\leq\min\{\boldsymbol{\pi}_{k,i}(\theta),\boldsymbol{\pi}_{\ell,i}(\theta)\}
\end{align}
where the last inequality follows from the fact that $0\leq\boldsymbol{\pi}_{k,i}(\theta),\boldsymbol{\pi}_{\ell,i}(\theta)\leq1$. Then, from \eqref{weights_definition} we have
\begin{align}
&\boldsymbol{a}_{\ell k,i}=\sum_{\theta\in\Theta}\boldsymbol{a}_{\ell k,i}(\theta)\leq\sum_{\theta\in\Theta}\boldsymbol{\pi}_{k,i}(\theta)\boldsymbol{\pi}_{\ell,i}(\theta)
\end{align}
Let us denote
\begin{align}
    \widehat{\boldsymbol{\lambda}}_{k,i}(\theta)\triangleq\log\frac{\boldsymbol{\pi}_{k,i}(\theta)\boldsymbol{\pi}_{\ell,i}(\theta)}{\boldsymbol{\pi}_{k,i}(\theta^{(k)})\boldsymbol{\pi}_{\ell,i}(\theta^{(\ell)})},\quad\theta\notin\Theta^{\star}_k\cap\Theta^{\star}_{\ell}.
\end{align}
Now, we follow the same rationale as in the proof of Proposition \ref{rate_localbeliefs} and as in \cite{nedic2017fast}. Using \eqref{HMM_1} we have
\begin{align}
\label{llrb_tel2}
    &\widehat{\boldsymbol{\lambda}}_i(\theta)=\sum^i_{t=1}\log\frac{L_k(\boldsymbol{\zeta}_{k,i}|\theta)}{L_k(\boldsymbol{\zeta}_{k,i}|\theta^{(k)})}+\sum^i_{t=1}\log\frac{L_{\ell}(\boldsymbol{\zeta}_{\ell,i}|\theta)}{L_{\ell}(\boldsymbol{\zeta}_{\ell,i}|\theta^{(\ell)})}\nonumber\\
    &+\lambda_{k,0}(\theta)+\lambda_{\ell,0}(\theta)
\end{align}
where the last two terms $\lambda_{k,0}(\theta)=\lambda_{\ell,0}(\theta)=0$ due to Assumption \ref{prior_beliefs}. Taking expectations in \eqref{llrb_tel2} we have
\begin{align}
\label{fas}
    &\mathbb{E}\{\widehat{\boldsymbol{\lambda}}_i(\theta)\}=-i\sum_{m\in\{k,\ell\}}d_m(\theta)\nonumber\\
    &\leq-i\min_{\theta\notin\Theta^{\star}_k\cap\Theta^{\star}_{\ell}}\sum_{m\in\{k,\ell\}}d_m(\theta).
\end{align}
Next,  by Assumption \ref{non_emptysupport}, 
we have for all $\theta\notin\Theta^{\star}_k\cap\Theta^{\star}_{\ell}$
\begin{align}
\label{KL_bound2}
    2\log\alpha\leq\sum_{m\in\{k,\ell\}}d_m(\theta)\leq2\log \frac{1}{\alpha}.
\end{align}
Consider the sequence of random variables $\boldsymbol{\zeta}_{1:i}=(\boldsymbol{\zeta}_{k,1},\boldsymbol{\zeta}_{\ell,1},\ldots,\boldsymbol{\zeta}_{k,i},\boldsymbol{\zeta}_{\ell,i})$. We want to establish that $\widehat{\boldsymbol{\lambda}}_{k,i}(\theta)$, which is a function of $\boldsymbol{\zeta}_{1:i}$, has bounded differences. For all $t$ such that $1\leq t\leq i$ we have
\begin{align}
\label{McD_bound}
    &\max_{\zeta_{k,t},\zeta_{\ell,t}}\widehat{\boldsymbol{\lambda}}_i(\theta)-\min_{\zeta_{k,t},\zeta_{\ell,t}}\widehat{\boldsymbol{\lambda}}_i(\theta)\nonumber\\
    &=\max_{\zeta_{k,t},\zeta_{\ell,t}}\left\{\log\frac{L_k(\zeta_{k,t}|\theta)}{L_k(\zeta_{k,t}|\theta^{(k)})}+\log\frac{L_{\ell}(\zeta_{\ell,t}|\theta)}{L_{\ell}(\zeta_{\ell,t}|\theta^{(\ell)})}\right\}\nonumber\\
    &-\min_{\zeta_{k,t},\zeta_{\ell,t}}\left\{\log\frac{L_k(\zeta_{k,t}|\theta)}{L_k(\zeta_{k,t}|\theta^{(k)})}+\log\frac{L_{\ell}(\zeta_{\ell,t}|\theta)}{L_{\ell}(\zeta_{\ell,t}|\theta^{(\ell)})}\right\}\nonumber\\
    &\leq2\log\frac{1}{\alpha}-2\log\alpha=4\log\frac{1}{\alpha}
\end{align}
where we utilized \eqref{KL_bound2}. Thus, $\widehat{\boldsymbol{\lambda}}_i(\theta)$ has bounded differences and as a result, we can apply McDiarmid's inequality. Then, by utilizing \eqref{McD} for $\rho_t=4\log\frac{1}{\alpha}$ (the bound from \eqref{McD_bound}) we obtain
\begin{align}
\label{bound_llrb2}
    \mathbb{P}\Big(\widehat{\boldsymbol{\lambda}}_i(\theta)-\mathbb{E}\{\widehat{\boldsymbol{\lambda}}_i(\theta)\}\geq\epsilon\Big)\leq\exp\Big(-\frac{2\epsilon^2}{16i(\log\frac{1}{\alpha})^2}\Big).
\end{align}
Then, since $\boldsymbol{\pi}_{k,i}(\theta),\boldsymbol{\pi}_{\ell,i}(\theta)\in(0,1)$ for all $\theta\notin\Theta^{\star}_k\cap\Theta^{\star}_{\ell}$ we have
\begin{align}
\label{in12}
    &\boldsymbol{a}_{\ell k,i}(\theta)=\boldsymbol{\pi}_{k,i}(\theta)\boldsymbol{\pi}_{\ell,i}(\theta)\leq\frac{\boldsymbol{\pi}_{k,i}(\theta)\boldsymbol{\pi}_{\ell,i}(\theta)}{\boldsymbol{\pi}_{k,i}(\theta^{(k)})\boldsymbol{\pi}_{\ell,i}(\theta^{(\ell)})}\nonumber\\
    &=\exp\Big(\widehat{\boldsymbol{\lambda}}_i(\theta)\Big).
\end{align}
Thus, for an arbitrary $\varepsilon$ we have
\begin{align}
\label{bound_belief2}
    &\mathbb{P}\Big(\boldsymbol{a}_{\ell k,i}(\theta)\geq\exp(\varepsilon)\Big)\nonumber\\
    &\leq\mathbb{P}\Big(\exp(\widehat{\boldsymbol{\lambda}}_i(\theta))\geq\exp(\varepsilon)\Big)=\mathbb{P}\Big(\widehat{\boldsymbol{\lambda}}_i(\theta)\geq\varepsilon\Big)\nonumber\\
    &\leq\mathbb{P}\left(\widehat{\boldsymbol{\lambda}}_i(\theta)-\mathbb{E}\{\widehat{\boldsymbol{\lambda}}_i\}
    \overset{\eqref{fas}}\geq\varepsilon+i\min_{\theta\notin\Theta^{\star}_k\cap\Theta^{\star}_{\ell}}\sum_{m\in\{k,\ell\}}d_m(\theta)\right)
\end{align}
Utilizing \eqref{bound_llrb2}, \eqref{bound_belief2} and setting 
\begin{align}
\varepsilon=-\frac{i}{2}\left(\min_{\theta\notin\Theta^{\star}_k\cap\Theta^{\star}_{\ell}}\sum_{m\in\{k,\ell\}}d_m(\theta)\right)
\end{align}
we have
\begin{small}
\begin{align}
\label{concentration_inequality2}
    &\mathbb{P}\left(\boldsymbol{a}_{\ell k,i}(\theta)\geq\exp(\varepsilon)\right)\nonumber\\
    &\leq\mathbb{P}\left(\widehat{\boldsymbol{\lambda}}_i(\theta)-\mathbb{E}\{\widehat{\boldsymbol{\lambda}}_i(\theta)\}\geq\frac{i}{2}\min_{\theta\notin\Theta^{\star}_k\cap\Theta^{\star}_{\ell}}\sum_{m\in\{k,\ell\}}d_m(\theta)\right)\nonumber\\
    &\leq\exp\left(-\frac{\frac{1}{2}\left(i\min_{\theta\notin\Theta^{\star}_k\cap\Theta^{\star}_{\ell}}\sum_{m\in\{k,\ell\}}d_m(\theta)\right)^2}{16i(\log\frac{1}{\alpha})^2}\right)
    \nonumber\\
    &=\exp\left(-\frac{\min_{\theta\notin\Theta^{\star}_k\cap\Theta^{\star}_{\ell}}\sum_{m\in\{k,\ell\}}d^2_m(\theta)}{32(\log\alpha)^2}i\right)
\end{align}
\end{small}%
Now we will prove part $2)$ of Lemma. From Proposition \ref{rate_localbeliefs} we have that 
\begin{align}
    \mathbb{P}\left(\boldsymbol{\pi}_{k,i}(\theta)\geq \exp(-x_ki)\right)\leq\exp(-y_ki)
\end{align}
for all $\theta\notin\Theta^{\star}_k$ for some $x_k,y_k>0$. 
Consider the events $Y_{\theta}=\{\boldsymbol{\pi}_{k,i}(\theta)<\exp(-x_ki)\}$ and $Y=\{\sum_{\theta\notin\Theta^{\star}_k}\boldsymbol{\pi}_{ k,i}(\theta)\geq\sum_{\theta\notin\Theta}\exp(-x_ki)=|\Theta\setminus\Theta^{\star}_k|\exp(-x_ki)\}$. Then, we have
\begin{align}
\label{in_sum1}
    \underset{\theta\in\Theta}\cap Y_{\theta}\Rightarrow\neg Y
\end{align}
which implies
\begin{align}
\label{in_sum2}
    Y\Rightarrow\underset{\theta\in\Theta}\cup\neg Y_{\theta}
\end{align}
where $\neg$ stands for negation. 
Then, \eqref{in_sum2} implies
\begin{align}
\label{cin_sum}
    &\mathbb{P}\left(\sum_{\theta\notin\Theta^{\star}_k}\boldsymbol{\pi}_k(\theta)\geq |\Theta\setminus\Theta^{\star}_k|\exp(-x_ki)\right)\nonumber\\
    &\leq\mathbb{P}\left(\cup_{\theta\notin\Theta^{\star}_k}\{\boldsymbol{\pi}_{k,i}(\theta)\geq \exp(-x_ki)\}\right)\nonumber\\
    &\leq|\Theta\setminus\Theta^{\star}_k|\exp(-y_ki)
\end{align}
where in the last inequality the union bound was utilized. Then, since $\boldsymbol{\pi}_{k,i}$ is a probability vector, we have
\begin{align}
\label{pr_vec}
    \sum_{\theta\in\Theta}\boldsymbol{\pi}_{k,i}(\theta)=\sum_{\theta\in\Theta^{\star}_k}\boldsymbol{\pi}_{k,i}(\theta)+\sum_{\theta\notin\Theta^{\star}_k}\boldsymbol{\pi}_{k,i}(\theta)=1.
\end{align}
Then, from \eqref{cin_sum} and \eqref{pr_vec} we obtain
\begin{align}
\label{cin_sumset}
    &\mathbb{P}\left(\sum_{\theta\in\Theta^{\star}_k}\boldsymbol{\pi}_{k,i}(\theta)\leq1-|\Theta\setminus\Theta^{\star}_k|\exp(-x_ki)\right)\nonumber\\
    &\leq|\Theta\setminus\Theta^{\star}_k|\exp(-y_ki)
\end{align}
Utilizing \eqref{equal_beliefs}, \eqref{cin_sumset} yields
\begin{align}
\label{in_auxil1}
    &\mathbb{P}\left(\boldsymbol{\pi}_{k,i}(\theta)\leq\frac{1}{|\Theta^{\star}_k|}-\frac{|\Theta\setminus\Theta^{\star}_k|}{|\Theta^{\star}_k|}\exp(-x_ki)\right)\nonumber\\
    &\leq|\Theta\setminus\Theta^{\star}_k|\exp(-y_ki),\quad\theta\in\Theta^{\star}_k.
\end{align}
Also, due to \eqref{equal_beliefs} and the fact that $\boldsymbol{\pi}_{k,i}$ is a probability vector, the probability mass placed on every $\theta\in\Theta^{\star}_k$ cannot exceed $1/|\Theta^{\star}_k|$. Thus, the following is true:
\begin{align}
\label{in_auxil2}
    &\mathbb{P}\left(\boldsymbol{\pi}_{k,i}(\theta)>\frac{1}{|\Theta^{\star}_k|}+\frac{|\Theta\setminus\Theta^{\star}_k|}{|\Theta^{\star}_k|}\exp(-x_ki)\right)\nonumber\\
    &=0\leq|\Theta^{\star}_k|\exp(-y_ki),\quad\theta\in\Theta^{\star}_k.
\end{align}
Combining \eqref{in_auxil1} and \eqref{in_auxil2}, we obtain
\begin{align}
    &\mathbb{P}\left(|\boldsymbol{\pi}_{k,i}(\theta)-\frac{1}{|\Theta^{\star}_k|}|\geq\frac{|\Theta\setminus\Theta^{\star}_k|}{|\Theta^{\star}_k|}\exp(-x_ki)\right)\nonumber\\
    &\leq|\Theta\setminus\Theta^{\star}_k|\exp(-y_ki),\quad\theta\in\Theta^{\star}_k,
\end{align}
for all $k\in\mathcal{N}$. Also, for all $\theta\notin\Theta^{\star}_k$ we can write ({\em see} Proposition \ref{rate_localbeliefs}):
\begin{align}
    \mathbb{P}\left(|\boldsymbol{\pi}_{k,i}(\theta)-0|\geq \exp(-x_ki)\right)\leq\exp(-y_ki)
\end{align}
Now, we can use the properties shown in Lemma \ref{in_form} ({\em see} Appendix \ref{ap_aux}) to show that $\boldsymbol{a}_{\ell k,i}(\theta)$ converges exponentially fast to $a_{\ell k,\infty}(\theta)$ for all $\theta$. From \eqref{combination_weights}, we see that the numerator of $\boldsymbol{a}_{\ell k,i}(\theta)$ converges exponentially fast to $\pi_{k,\infty}(\theta)\pi_{\ell,\infty}(\theta)$ from \eqref{f3}. The denominator (i.e., $\boldsymbol{\sigma}_{k,i}$) also converges exponentially fast because we can repeatedly apply \eqref{f2} and \eqref{f3}, as $\boldsymbol{\sigma}_{k,i}$ is comprised of products and sums of random variables satisfying \eqref{f2}, \eqref{f3}.  The inverse of the denominator also converges exponentially fast due to \eqref{f1} (note that $\boldsymbol{\sigma}_{k,i}\geq1$ for all $i\geq1$). Finally, we apply \eqref{f3} by setting $\boldsymbol{x}_i=\boldsymbol{\pi}_{k,i}(\theta)\boldsymbol{\pi}_{\ell,i}(\theta)$ and $\boldsymbol{y}_i= \boldsymbol{\sigma}^{-1}_{k,i}$. 
Thus, we get the statement of the Lemma.
\end{proof}
\bibliographystyle{IEEEtran}
\bibliography{references}

\begin{thebibliography}{10}
\providecommand{\url}[1]{#1}
\csname url@samestyle\endcsname
\providecommand{\newblock}{\relax}
\providecommand{\bibinfo}[2]{#2}
\providecommand{\BIBentrySTDinterwordspacing}{\spaceskip=0pt\relax}
\providecommand{\BIBentryALTinterwordstretchfactor}{4}
\providecommand{\BIBentryALTinterwordspacing}{\spaceskip=\fontdimen2\font plus
\BIBentryALTinterwordstretchfactor\fontdimen3\font minus
  \fontdimen4\font\relax}
\providecommand{\BIBforeignlanguage}[2]{{%
\expandafter\ifx\csname l@#1\endcsname\relax
\typeout{** WARNING: IEEEtran.bst: No hyphenation pattern has been}%
\typeout{** loaded for the language `#1'. Using the pattern for}%
\typeout{** the default language instead.}%
\else
\language=\csname l@#1\endcsname
\fi
#2}}
\providecommand{\BIBdecl}{\relax}
\BIBdecl

\bibitem{lalitha2014social}
A.~Lalitha, A.~Sarwate, and T.~Javidi, ``Social learning and distributed
  hypothesis testing,'' in \emph{Proc. IEEE International Symposium on
  Information Theory}, Honolulu, Hawaii, 2014, pp. 551--555.

\bibitem{lalitha2018social}
A.~Lalitha, T.~Javidi, and A.~D. Sarwate, ``Social learning and distributed
  hypothesis testing,'' \emph{IEEE Transactions on Information Theory},
  vol.~64, no.~9, pp. 6161--6179, 2018.

\bibitem{nedic2017fast}
A.~Nedi{\'c}, A.~Olshevsky, and C.~A. Uribe, ``Fast convergence rates for
  distributed non-bayesian learning,'' \emph{IEEE Transactions on Automatic
  Control}, vol.~62, no.~11, pp. 5538--5553, 2017.

\bibitem{zhao2012learning}
X.~Zhao and A.~H. Sayed, ``Learning over social networks via diffusion
  adaptation,'' in \emph{Proc. Asilomar Conference on Signals, Systems and
  Computers}, 2012, pp. 709--713.

\bibitem{bordignon2020social}
V.~Bordignon, V.~Matta, and A.~H. Sayed, ``Social learning with partial
  information sharing,'' in \emph{Proc. IEEE International Conference on
  Acoustics, Speech and Signal Processing (ICASSP)}, Barcelona, Spain, 2020,
  pp. 5540--5544.

\bibitem{jadbabaie2012non}
A.~Jadbabaie, P.~Molavi, A.~Sandroni, and A.~Tahbaz-Salehi, ``Non-bayesian
  social learning,'' \emph{Games and Economic Behavior}, vol.~76, no.~1, pp.
  210--225, 2012.

\bibitem{molavi2018theory}
P.~Molavi, A.~Tahbaz-Salehi, and A.~Jadbabaie, ``A theory of non-bayesian
  social learning,'' \emph{Econometrica}, vol.~86, no.~2, pp. 445--490, 2018.

\bibitem{mitra2020new}
A.~Mitra, J.~A. Richards, and S.~Sundaram, ``A new approach to distributed
  hypothesis testing and non-bayesian learning: Improved learning rate and
  byzantine-resilience,'' \emph{IEEE Transactions on Automatic Control},
  vol.~66, no.~9, pp. 4084--4100, 2020.

\bibitem{ntemos2021social}
K.~Ntemos, V.~Bordignon, S.~Vlaski, and A.~H. Sayed, ``Social learning under
  inferential attacks,'' in \emph{IEEE International Conference on Acoustics,
  Speech and Signal Processing (ICASSP)}.\hskip 1em plus 0.5em minus
  0.4em\relax IEEE, 2021, pp. 5479--5483.

\bibitem{bordignon2021network}
V.~Bordignon, S.~Vlaski, V.~Matta, and A.~H. Sayed, ``Network classifiers based
  on social learning,'' in \emph{Proc. IEEE International Conference on
  Acoustics, Speech and Signal Processing (ICASSP)}, Toronto, Canada, 2021, pp.
  5185--5189.

\bibitem{chen2015adaptive}
J.~Chen, C.~Richard, and A.~H. Sayed, ``Adaptive clustering for multitask
  diffusion networks,'' in \emph{Proc. European Signal Processing Conference
  (EUSIPCO)}, Nice, France, 2015, pp. 200--204.

\bibitem{chen2015diffusion}
------, ``Diffusion {L}{M}{S} over multitask networks,'' \emph{IEEE
  Transactions on Signal Processing}, vol.~63, no.~11, pp. 2733--2748, 2015.

\bibitem{plata2017heterogeneous}
J.~Plata-Chaves, A.~Bertrand, M.~Moonen, S.~Theodoridis, and A.~M. Zoubir,
  ``Heterogeneous and multitask wireless sensor networks—algorithms,
  applications, and challenges,'' \emph{IEEE Journal of Selected Topics in
  Signal Processing}, vol.~11, no.~3, pp. 450--465, 2017.

\bibitem{zhao2015distributed}
X.~Zhao and A.~H. Sayed, ``Distributed clustering and learning over networks,''
  \emph{IEEE Transactions on Signal Processing}, vol.~63, no.~13, pp.
  3285--3300, 2015.

\bibitem{nassif2020multitask}
R.~Nassif, S.~Vlaski, C.~Richard, J.~Chen, and A.~H. Sayed, ``Multitask
  learning over graphs: An approach for distributed, streaming machine
  learning,'' \emph{IEEE Signal Processing Magazine}, vol.~37, no.~3, pp.
  14--25, 2020.

\bibitem{teklehaymanot2015robust}
F.~K. Teklehaymanot, M.~Muma, B.~B{\'e}jar, P.~Binder, A.~Zoubir, and
  M.~Vetterli, ``Robust diffusion-based unsupervised object labelling in
  distributed camera networks,'' in \emph{Proc. AFRICON 2015}.\hskip 1em plus
  0.5em minus 0.4em\relax IEEE, Addis Ababa, Ethiopia, 2015, pp. 1--6.

\bibitem{bertrand2010distributed}
A.~Bertrand and M.~Moonen, ``Distributed adaptive node-specific signal
  estimation in fully connected sensor networks—part i: Sequential node
  updating,'' \emph{IEEE Transactions on Signal Processing}, vol.~58, no.~10,
  pp. 5277--5291, 2010.

\bibitem{bertrand2010distributed2}
------, ``Distributed adaptive node-specific signal estimation in fully
  connected sensor networks—part ii: Simultaneous and asynchronous node
  updating,'' \emph{IEEE Transactions on Signal Processing}, vol.~58, no.~10,
  pp. 5292--5306, 2010.

\bibitem{plata2015distributed}
J.~Plata-Chaves, N.~Bogdanovi{\'c}, and K.~Berberidis, ``Distributed
  diffusion-based {L}{M}{S} for node-specific adaptive parameter estimation,''
  \emph{IEEE Transactions on Signal Processing}, vol.~63, no.~13, pp.
  3448--3460, 2015.

\bibitem{bogdanovic2014distributed}
N.~Bogdanovi{\'c}, J.~Plata-Chaves, and K.~Berberidis, ``Distributed
  incremental-based {L}{M}{S} for node-specific adaptive parameter
  estimation,'' \emph{IEEE Transactions on Signal Processing}, vol.~62, no.~20,
  pp. 5382--5397, 2014.

\bibitem{marano2021decision}
S.~Marano and A.~H. Sayed, ``Decision learning and adaptation over multi-task
  networks,'' \emph{IEEE Transactions on Signal Processing}, vol.~69, pp.
  2873--2887, 2021.

\bibitem{Sayed_2014}
A.~H. Sayed, ``Adaptation, learning, and optimization over networks,''
  \emph{Foundations and Trends® in Machine Learning}, vol.~7, no. 4-5, pp.
  311--801, 2014.

\bibitem{nedic2015nonasymptotic}
A.~Nedi{\'c}, A.~Olshevsky, and C.~A. Uribe, ``Nonasymptotic convergence rates
  for cooperative learning over time-varying directed graphs,'' in \emph{Proc.
  2015 American Control Conference (ACC)}, Chicago, Illinois, 2015, pp.
  5884--5889.

\bibitem{matta2019graph}
V.~Matta, A.~Santos, and A.~H. Sayed, ``Graph learning with partial
  observations: Role of degree concentration,'' in \emph{Proc. IEEE
  International Symposium on Information Theory (ISIT)}, Paris, France, 2019,
  pp. 1312--1316.

\bibitem{doob1940regularity}
J.~L. Doob, ``Regularity properties of certain families of chance variables,''
  \emph{Transactions of the American Mathematical Society}, vol.~47, no.~3, pp.
  455--486, 1940.

\bibitem{horn2012matrix}
R.~A. Horn and C.~R. Johnson, \emph{Matrix Analysis}.\hskip 1em plus 0.5em
  minus 0.4em\relax Cambridge University Press, 2012.

\bibitem{bordignon2020socialarxiv}
\BIBentryALTinterwordspacing
V.~Bordignon, V.~Matta, and A.~H. Sayed, ``Social learning with partial
  information sharing,'' \emph{arXiv:2006.13659}, 2020. [Online]. Available:
  \url{https://arxiv.org/abs/2006.13659}
\BIBentrySTDinterwordspacing

\bibitem{hurley2014concise}
P.~J. Hurley, \emph{A Concise Introduction to Logic}.\hskip 1em plus 0.5em
  minus 0.4em\relax Cengage Learning, 2014.

\end{thebibliography}
\end{document}